\documentclass[preprint,number]{elsarticle}
\usepackage[english]{babel}
\usepackage{graphicx}
\usepackage[linesnumbered,ruled,vlined]{algorithm2e}
\usepackage{amsfonts,amsmath}
\usepackage{multirow}
\usepackage{amsmath}
\usepackage{color}
\allowdisplaybreaks

\newtheorem{theorem}{Theorem}
\newtheorem{lemma}[theorem]{Lemma}
\newtheorem{observation}{Observation}
\newtheorem{corollary}[theorem]{Corollary}
\newdefinition{rmk}{Remark}
\newproof{proof}{Proof}
\newdefinition{definition}{Definition}
\newdefinition{example}{Example}

\begin{document}

\begin{frontmatter}

\title{Task Assignment in Tree-Like Hierarchical Structures}

\author[ieu]{Cem Evrendilek}
\ead{cem.evrendilek@ieu.edu.tr}
\author[metu]{Ismail Hakki Toroslu\corref{cor}}
\ead{toroslu@ceng.metu.edu.tr}
\author[metu]{Sasan Hashemi}
\ead{sasan@ceng.metu.edu.tr}

\cortext[cor]{Corresponding author. Tel.: +90 312 210 5585}
\address[ieu]{\.{I}zmir University of Economics\\Computer Engineering Department\\35330 \.{I}zmir, Turkey}
\address[metu]{Middle East Technical University\\Computer Engineering Department\\06531 Ankara, Turkey}

\begin{abstract}
Most large organizations, such as corporations, are hierarchical organizations. In hierarchical organizations
each entity in the organization, except the root entity, is a sub-part of another entity.  In this paper we
study the task assignment problem to the entities of tree-like hierarchical organizations. The inherent
tree structure introduces an interesting and challenging constraint to the standard assignment problem.
When a task is assigned to an entity in a hierarchical organization, the whole entity, including its
sub-entities, is responsible from the execution of that particular task. In other words,
if an entity has been assigned to a task, neither its descendants nor its ancestors can be
assigned to a task. Sub-entities cannot be assigned as they have an ancestor already occupied.
Ancestor entities cannot be assigned since one of their sub-entities has already been employed
in an assignment. In the paper, we formally introduce this new version of the assignment problem
called Maximum Weight Tree Matching ($MWTM$), and show its NP-hardness. We also propose
an effective heuristic solution based on an iterative LP-relaxation to it.  
\end{abstract}


\end{frontmatter}

\section{Introduction}

In the standard assignment problem (or as sometimes referred to linear assignment problem) \cite{BDM09},
the number of tasks and the number of agents are equal, and a scalar value is used
to represent the cost/performance of assigning a task to an agent.  The objective of the assignment
problem is to determine an assignment such that each task is assigned to a different agent and the summation
of the costs/profits of the assignment is minimized/maximized. Many different variations of this problem
have already been studied including \emph{Generalized Assignment Problem} \cite{CKR06, FGMS06, G93, A95}.
In this work, we also investigate a new version of the standard assignment problem
which appears in real-life applications. 

In real-life, most of the large organizations such as corporations, governments, military etc.,
have hierarchical structures. Hierarchical organizations are nothing but trees where each
node corresponds to an entity in the organization, and entity sub-entity relationships are
represented as parent-child relationships.

In the standard assignment problem, agents are flat, and have got no structure
imposed on them one task is assigned to one agent. However, in Maximum
Weight Tree Matching ($MWTM$) problem, since agents are
organized as a tree, and sub-entities in the tree represent sub-parts of the agents,
an additional constraint, named hereafter as \emph{hierarchy constraint},
 is introduced to the assignment problem: When a task is assigned
to an agent, no other assignment can be made to its sub-entities, as they are assumed to
be a part of an agent already assigned. This constraint indirectly implies another
constraint. Since an agent should be assigned to a task as a whole along with its
sub-parts, when one of its sub-parts has already been assigned, then it cannot be
assigned itself to any task. In other words, if an agent is assigned to a task, none of
its ancestors in the tree can be assigned at all. In a more general term, on every
path from the root to a leaf in a tree, there could be at most a single assignment.
This, in turn, is easily seen to lead to the observation that the number of leaves
in the tree should be at least equal to the number of tasks to be executed.
Otherwise no feasible assignment exists.

A simpler version of $MWTM$ problem where each node has the same assignment weight
for all the tasks to be performed has been introduced in \cite{GT10pack}. It is
called \textquotedblleft tree like weighted set packing\textquotedblright \,in \cite{GT10pack}
since the set-subset relationships form a tree, and the weight assigned to each set
(or each node in the tree) can be interpreted as the weight of assigning a task to that
node. The same hierarchy (independence) constraint has been enforced to prevent
the selection of two sets having set-subset relationships (either directly or indirectly),
and finally the number of sets to be packed (selected) is given to maximize the total weight.
That problem effectively becomes a simpler version of the problem studied
in this paper, and an effective dynamic programming solution to it has been
developed in \cite{GT10pack}.

Although many different versions of assignment problems have been defined and explored,
there are only a very few problems remotely related to $MWTM$ problem reported in the
literature, such as \cite{WZS13}, and \cite{SRSP06}. Similar to $MWTM$, both of these
problems introduce different kinds of set constraints on the vertices of a bipartite graph,
and they have both been shown to be NP-hard. Therefore, heuristic solutions have been
proposed, namely a greedy heuristic for \cite{WZS13}, and a genetic algorithm based
solution for \cite{SRSP06}, and these solutions have been shown to be quite effective.

$MWTM$ problem has already been introduced in \cite{GT10}, and
a generic heuristic (genetic algorithm) has been used to solve it.
In \cite{GT10}, it has been shown that GA works quite effectively in terms of solution
quality for randomly generated inputs. Although the number of iterations were not
very large, due to the cost of each genetic operator among the chromosome populations,
each iteration takes a considerably long time to complete, and therefore we have
observed that the execution times are much higher to reach to the level of near-optimal
results obtained with the approach proposed in this paper. Since GA approach uses a
generic heuristic (slightly customized for the problem), it is actually not fair to compare
it with our problem-specific heuristic, which is much more effective.
Moreover, although GA approach has been applied to different sized inputs,
significant input parameters have not been explored in its evaluation
in \cite{GT10} corresponding to the structure and the distributions of the weights.
That is why we have compared the quality of the solutions of our heuristic proposed
in this paper with that of ILP only which produces the optimal (whenever possible). 
This paper has the following additional contributions to \cite{GT10}:
\begin{itemize}
\item{The problem is shown to be NP-hard,}
\item{An Integer Linear Programming (ILP) model of the problem is given,}
\item{An iterative Linear Programming (LP) relaxation solution is developed,}
\item{The effectiveness of the proposed iterative LP-relaxation solution is verified through extensive tests.}
\end{itemize}

Iterative LP-relaxation or rounding algorithms have previously been used.
A factor $2$ approximation algorithm is presented in \cite{J01} for finding a minimum-cost subgraph
having at least a specified number of edges in each cut. This class of problems defined
in \cite{J01} includes the generalized Steiner network problem also known as the survivable
network design problem. The algorithm in \cite{J01} first solves the linear relaxation of
ILP formulation of the problem, and then iteratively rounds off the solution. The approach
taken in \cite{J01} has been generalized and formalized in \cite{JThesis00}.
In order to exploit the full power of LP, a new technique called iterative rounding has been
introduced in \cite{JThesis00}. Iterative rounding is used in \cite{JThesis00} to
iteratively recompute the best fractional solution while maintaining the
rounding of the previous phases. Although an iterative rounding based heuristic
solution is developed in this paper for $MWTM$, the presence of the hierarchy
constraint does not simply lend itself to the consideration of fractional values
from the highest to the smallest.

The rest of the paper has been organized as follows. The next section formally introduces the problem,
and proves its NP-hardness. Section~\ref{sec:ILP}, describes a mathematical (integer linear programming)
formulation, and Section~\ref{sec:bottom-up} presents how its relaxation to LP can be iteratively used as an
effective heuristic. Section~\ref{sec:experiments} describes the experiments and their results.
Finally, the last section presents concluding remarks.

\section{Problem Description and its NP-Hardness}\label{sec:np-hardness}

We will now introduce \textit{Maximum Weight Tree Matching} ($MWTM$) problem formally.
\begin{definition}\label{def:MWTM}
A tree $T$ with $n$ nodes rooted at a node $r$, and a separate set of $m$ tasks are given.
Associated with each node $i$ in $T$ is a real valued function $w_{i,j}$ denoting the
weight of assigning node $i$ to task $j$ for all $i \in \{1..n\}$ and $j \in \{1..m\}$
The problem of finding an assignment of all tasks to nodes in $T$ with the maximum
total weight in such a way that the assignment between nodes and tasks forms
a matching, and no node assigned to a task is allowed to have any ancestors (or descendants)
which have also been assigned to a task is named \textit{MWTM}.
\end{definition}

It should be noted that the requirement for the weight function to
be defined for all combinations of nodes and tasks in $MWTM$ stems
from a deliberate decision. $MWTM$ is more restricted than its possible
variants where some combinations of nodes and tasks can be forbidden.
As $MWTM$ can be reduced directly to these more general forms,
NP-hardness of them would easily follow once $MWTM$ is shown to be NP-hard.

\begin{figure}[tbh]
	\centering
	\includegraphics[width=4.8in]{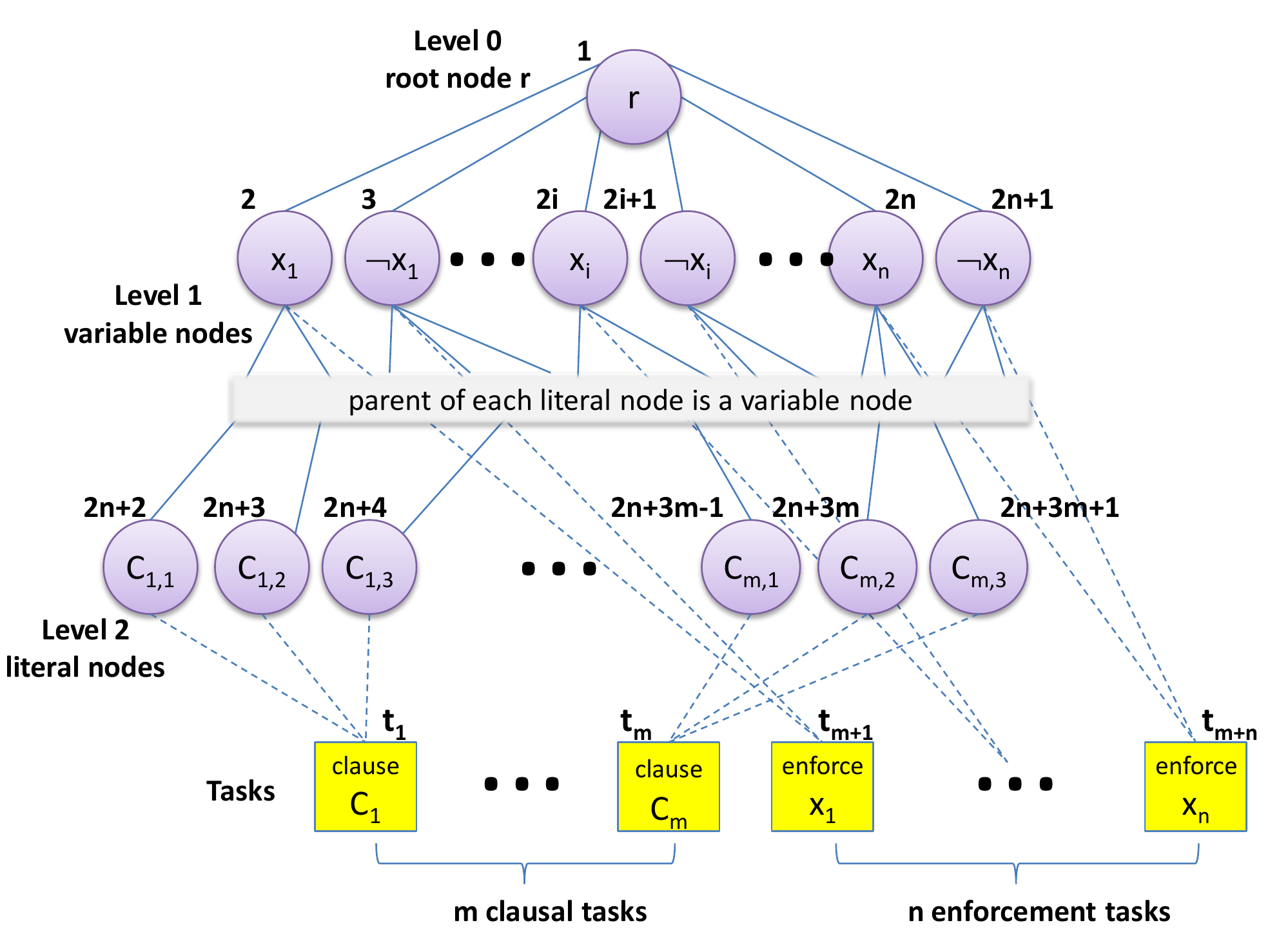}
	\caption{Transforming an E3-SAT instance to the corresponding instance of $MWTM$.
The solid lines between the nodes are the tree edges while the dashed lines between
the nodes and the tasks correspond to the weight function $w_{i,j}$.}
	\label{fig:reduction}
\end{figure}

The constraint associated with the hierarchical structure of the tree dictates that no
two nodes on the same path from the root $r$ to a leaf node in $T$ can ever be simultaneously
assigned in a solution to an instance of $MWTM$.
\begin{definition}\label{def:independentPath}
Two paths in a tree from the root to any two distinct nodes are said to be
\textit{independent paths} if and only if none of the two paths is a subset of the other.
\end{definition}
In the light of this definition, the hierarchy constraint can simply be restated
as the requirement that the paths from the assigned nodes to the root
are all pairwise independent.

$MWTM$ can be shown to be NP-hard by a polynomial time reduction from
E3-SAT which is a variant of 3-satisfiability (3-SAT) problem.
E3-SAT (resp. 3-SAT) is defined to be the problem of deciding whether a satisfying truth
assignment is possible for the variables of a given Boolean formula in
Conjunctive Normal Form (CNF) where each clause is a disjunction of
exactly (resp. at most) three literals each of which is either a variable or its negation.
3-SAT is one of Karp's $21$ NP-complete problems \cite{K72}. Any given instance
of 3-SAT can be easily transformed to a corresponding instance of E3-SAT by introducing
three new dummy variables, $d_1$, $d_2$, and $d_3$. While only $d_1$ is inserted
into the clauses with one literal, both $d_1$ and $d_2$ are inserted into the clauses
with two literals. In order to make sure in any satisfying assignment that the dummy
variables can only be set to false, the conjunction of all maxterms of the dummy
variables except $d_1 + d_2 + d_3$ are finally appended to the clauses each with
exactly three literals now. The NP-completeness of E3-SAT is hence confirmed.

A given instance of E3-SAT problem is transformed to a corresponding instance
of $MWTM$ in time polynomial in the size of the input Boolean expression. Let a given
instance of E3-SAT have $n$ variables denoted by $x_i$ where $i \in [1..n]$ and a
3-CNF formula $C_1 \wedge C_2 \wedge \ldots \wedge C_m$ where each $C_i$ represented
by $C_{i,1} \vee C_{i, 2} \vee C_{i, 3}$ is a disjunction of three literals corresponding
to either a variable or its negation. The transformation starts by introducing the root
node designated by $r$ to the initially empty tree $T$ of the corresponding $MWTM$
instance at level $0$. The root node $r$ is numbered as $1$.
For each variable $x_i$, two child nodes to root r are then created numbered $2i$
for $x_i$, and $2i+1$ for $\neg x_i$ corresponding to assigning $true$ and $false$
respectively to this variable. The parents of all such nodes are set to point
to node $r$. As there are $n$ distinct variables in the given E3-SAT instance,
the root $r$ of $T$ in the corresponding $MWTM$ instance becomes populated
with a total of $2n$ children at level 1 of $T$ after this step. These are called
\textit{variable nodes} (see Figure~\ref{fig:reduction}).
In the final step of the construction of $T$, for each literal $C_{i, j}$ where $i \in [1..m]$, and
$j \in [1..3]$, a node numbered $1+2n+3(i-1)+j$ is created. The parent of such
a node is set to $2k$ if $C_{i, j} = x_k$, and to $2k+1$ otherwise if $C_{i, j} = \neg x_k$
where $k \in [1..n]$. What this step achieves in effect for each node corresponding to
assigning $true$ to $x_k$ or to its negation $\neg x_k$ at level $1$ is the creation
of as many children at the next level $2$ under the relevant variable node as there are
occurrences of the corresponding variable in the clauses of the given E3-SAT instance.
The tree $T$ constructed is shown in Figure~\ref{fig:reduction}.
While parent-child relationships are indicated by solid lines in this figure,
dashed lines depict the weight function $w_{i, j}$. It should be noted that
the variable nodes at level $1$ will have as many children as there are
occurrences of the corresponding literal at level $2$ which is implied by
the existence of multiple edges emanating from a variable node while the nodes
corresponding to literals in clauses at level $2$ will have a single edge to
their parent as shown in the figure. The nodes at level $2$ are accordingly
called \textit{literal nodes}.

Once we obtain the tree $T$ in $MWTM$ instance corresponding to the given instance of
E3-SAT, we also set the number of tasks to $m + n$. Each task $t_i$ for $i \in [1..m]$
corresponds to satisfying a clause $C_i$. We call these \textit{clausal tasks}. Each task
$t_i$ for $i \in [m+1..m+n]$ among the rest of the tasks , however, are used to enforce
that the corresponding variable $x_{i-m}$ is set to either one of $true$ or $false$
consistently over all clauses. We call such tasks \textit{enforcement tasks}.

Apparently, the total number of nodes in $T$ in the corresponding instance of $MWTM$ is
given by $1 + 2n + 3m$ where $n$ and $m$ are the number of variables and clauses
respectively specified in the given E3-SAT instance. The number of tasks, on the other hand,
is $n + m$. The concluding step of the transformation is to appropriately set the corresponding
values $w_{i, j}$ for all nodes $i \in [1..1+2n+3m]$ in $T$, and all tasks $j \in [1..m+n]$
as shown in Equation~\ref{eqn:weight} below:

\begin{equation}\label{eqn:weight}
w_{i, j}= 
	\left\{ 
	\begin{array}{rl} 
		0, & \text{if } i=1 \wedge j \in [1..m+n] \\
		1, & \text{if } i \in [2..2n+1] \wedge j = m + \lfloor \frac{i}{2} \rfloor \\
		1, & \text{if } i \in [2n+2..2n+3m+1] \wedge j = \lfloor \frac{i-2n-2}{3} \rfloor + 1 \\
		0, & \text{otherwise }
	\end{array}	
	\right.
\end{equation}
The weights of carrying out any one task by the root node are all initialized to zero.
For a variable node $i \in [2..2n+1]$ at level $1$ corresponding to
$x_{\lfloor \frac{i}{2} \rfloor}$ or $\neg x_{\lfloor \frac{i}{2} \rfloor}$
depending on whether $i$ is even or odd respectively, however, the weights of executing tasks are set
in such a way that a consistent assignment of truth values to individual variables can be enforced.
The only task whose execution by node $i$ can have a positive contribution to the solution is
therefore the corresponding enforcement task $t_{m + \lfloor \frac{i}{2} \rfloor}$. At level $2$
are the literal nodes ranging from $2n+2$ to $1+2n+3m$ corresponding to the literals in the
clauses of the given E3-SAT instance. Since each literal can accordingly be set to satisfy a clause
in which it occurs, the weight $w_{i, j}$ of assigning a level $2$ node $i$ representing a literal
$C_{p, q}$ to clausal task $t_j$ corresponding to the clause $C_p$ itself is appropriately set to $1$
to reflect a feasible assignment. Therefore, the equalities $i=1+2n+3(p-1)+q$, and
$j=p$ must hold. Noting that $q$ can only take on the values $1$ through $3$
inclusive readily gives $p = \lfloor \frac{i-2n-2}{3} \rfloor + 1$, and $q = (i-2n-2) \mod 3 + 1$.
All other combinations of nodes and tasks have weight $0$.

It should be pointed out that an $MWTM$ instance so constructed would always lend
itself to feasible solutions since the number of leaf nodes in $T$ is greater than or equal
to the number of tasks. This last inequality can be seen to hold by noting that $m \ge (2n-t)/3$
where $t \in [0..n)$ denotes the number variable nodes without any children in $T$ based
on the assumption that at least one of a variable or its negation is used in one of $m$ clauses
in the given E3-SAT instance. As $m$, $n$, and $t$ are all non-negative,
$m \ge (2n-t)/3= \frac{2}{3}n-\frac{1}{3}t \ge \frac{1}{2}n-\frac{1}{3}t \ge \frac{1}{2}n-\frac{1}{2}t$
is easily obtained. Multiplying both sides of $m \ge \frac{1}{2}n-\frac{1}{2}t$ by two, and then adding
$m$ to both sides, we obtain $3m \ge m+n-t$, and then $3m+t \ge m+n$ by rearranging.
The feasibility of the corresponding $MWTM$ instances obtained through the transformation
described are hence confirmed.

Given the transformation described, we make the following straightforward observation
to be used in a lemma to follow.

\begin{observation}\label{obs:varOrNegation} In any solution with total weight $n+m$ to the corresponding $MWTM$
instance obtained from a given E3-SAT instance through the transformation described, a literal
node at level $2$ can be assigned to a related clausal task iff no other literal node corresponding
to its negation at the same level has already been allocated.
\end{observation}

\begin{proof}
In any solution with total weight $n+m$ to the corresponding $MWTM$ instance after the transformation
depicted, all the enforcement tasks should have already been assigned to variable nodes at level $1$ in $T$.
Only the children at level $2$ of the unassigned variable nodes at level $1$ now, by the definition of
$MWTM$, can be used to fulfill the tasks corresponding to the clauses which ensure a consistent
assignment.
\qed
\end{proof}

The following lemma can now be proved easily.

\begin{lemma}\label{lem:iff} A given E3-SAT instance with $n$ variables, and $m$ clauses is satisfiable
iff the corresponding $MWTM$ instance obtained through the transformation described above has a
solution with total weight $n+m$.
\end{lemma}

\begin{proof}
Let us first prove the sufficiency part:  If a given E3-SAT instance is satisfiable then there exists an
assignment of truth values to all $n$ variables which makes all $m$ clauses evaluate to true.
This, in turn, implies that at least one literal in every clause can be made true. The corresponding
$MWTM$ instance is then easily seen to have an optimal assignment with weight $n+m$:
Each task corresponding to a clause in this scheme is assigned to one node at level $2$
corresponding to one of the literals satisfying this clause while each enforcement task is
assigned to the node at level $1$ representing the negation of the literal evaluating to
true in a satisfying truth assignment to the given E3-SAT instance.
This is indeed a matching since each task is matched to a different node and no node which
is a parent of an already assigned literal node is assigned to a task. The latter is guaranteed by
the fact that if a literal node is assigned to a clausal task, then the node corresponding to
the negation of this literal at a higher level can only be used to accomplish the respective
enforcement task. The total weight is also the maximum possible as no weight value can
be greater than $1$.

In order to prove the necessity part, let us assume that there exists a solution with total
weight $n+m$ to the corresponding $MWTM$ instance. A truth assignment for the
given E3-SAT instance can be obtained by setting each variable $x_i$ to true if the
corresponding enforcement task $m+i$ is assigned to node $2i+1$, and to false if
the assignment is to node $2i$. This truth assignment definitely satisfies 3-CNF
expression of the given E3-SAT instance by Observation~\ref{obs:varOrNegation} above.
\qed
\end{proof}

To illustrate the idea in the reduction process, let us consider the following
example.

\begin{example}
A 3-CNF formula $(p \vee \neg q \vee \neg p) \wedge (p \vee r \vee \neg s) \wedge (q \vee r \vee s)$
with $4$ variables, and $3$ clauses is given. While the variables are named $p$, $q$, $r$, and $s$,
the clauses are denoted by $C_1 = (p \vee \neg q \vee \neg p)$, $C_2 = (p \vee r \vee \neg s)$,
and $C_3 = (q \vee r \vee s)$. The corresponding tree structure obtained through
the transformation just described is given in Figure~\ref{fig:reductionExample},
while the accompanying weights of assigning nodes to tasks are shown
in Figure~\ref{fig:reductionExampleWeight}. While the nodes are numbered from
$1$ through $18$, tasks are called $t_{C1}$, $t_{C2}$, and $t_{C3}$
corresponding to the clausal tasks, and $t_p$, $t_q$, $t_r$, and $t_s$
corresponding to the enforcement tasks.

\begin{figure}[tbh]
	\centering
	\includegraphics[width=4.4in]{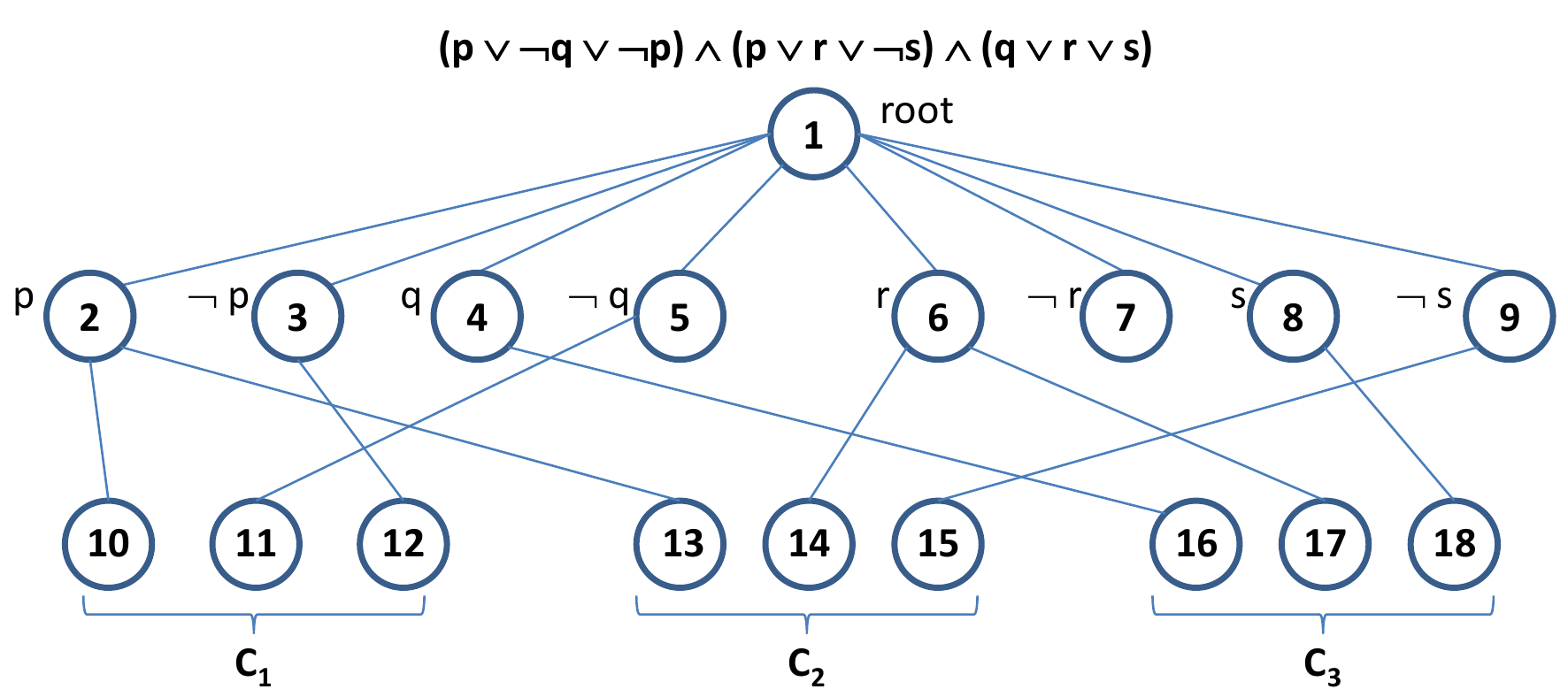}
	\caption{The tree for the corresponding $MWTM$ instance obtained from the example E3-SAT instance.}
	\label{fig:reductionExample}
\end{figure}

The given 3-CNF Boolean expression
is satisfiable if and only if $MWTM$ instance identified with the corresponding tree
structure given in Figure~\ref{fig:reductionExample}, and accompanying weight
function depicted in Figure~\ref{fig:reductionExampleWeight} has an assignment
with a total weight of $3+4=7$.

\begin{figure}[tbh]
	\centering
	\includegraphics[width=4.4in]{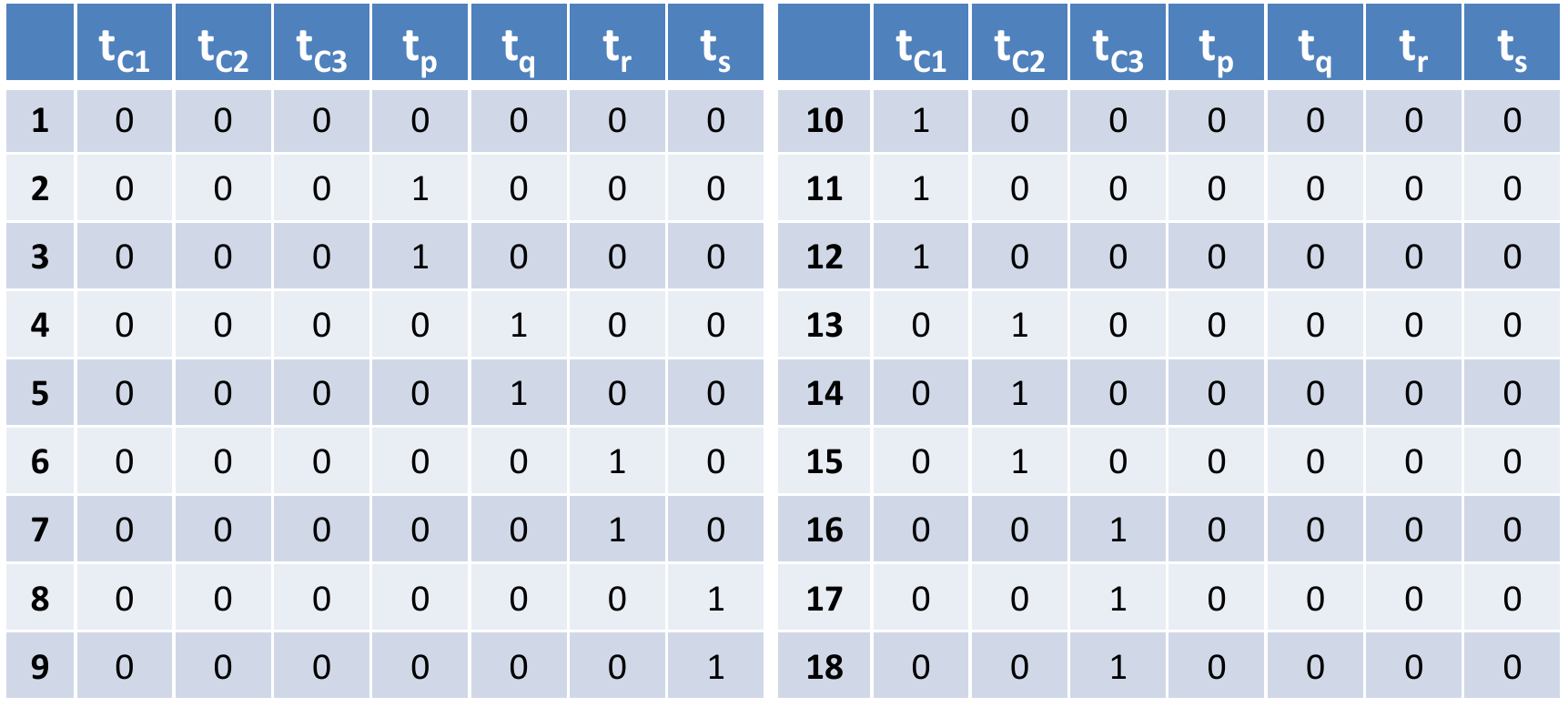}
	\caption{$w_{i, j}$ values for the corresponding $MWTM$ instance obtained from the example E3-SAT instance.}
	\label{fig:reductionExampleWeight}
\end{figure}

It should be noted that the construction constrains the weight values for variable nodes
$2$ through $9$ in the table to the left of Figure~\ref{fig:reductionExampleWeight}
through enforcement tasks in such a way that a node corresponding to a variable
and another node corresponding to its negation can not at the same time contribute
to a solution. An inspection of the weight values for literal nodes $10$ through $18$
in the table to the right in Figure~\ref{fig:reductionExampleWeight} reveals
similarly that only one of three such nodes can contribute to a solution
through a corresponding clausal task.

If a given 3-CNF formula is satisfiable, any satisfying truth assignment
induces a straightforward matching of the nodes to all the available tasks
such that the variable nodes at level $1$ evaluating to false
with respect to the given truth assignment are all assigned to their
corresponding enforcement tasks leaving only their negations for
a consistent instantiation over the entire set of clauses.
There are several ways to satisfy the above example
formula. Let us assume that we pick an assignment as follows: Variables
$p$ and $q$ are both assigned to $true$ while $r$ and $s$ can be assigned randomly.
 If we reflect these choices on the instance of $MWTM$ obtained, tasks
$t_{C1}$ and $t_{C2}$ are assigned respectively to nodes $10$ and $13$
both corresponding to $p$ while task $t_{C3}$ is assigned to node $16$
corresponding to $q$. Then, we can assign enforcement task $t_p$ to node
$3$ corresponding to $\neg p$, task $t_q$ to node $5$ corresponding to
$\neg q$. Finally, we assign task $t_r$ to either one of the nodes $6$ or $7$,
and task $t_s$ to either one of the nodes $8$ or $9$. For the last two, $r$ and $s$,
the choice is not really important, since neither one of these variables has been used
in satisfying the clauses.\qed
\end{example}

The transformation described in this section is certainly polynomial in the size of the given E3-SAT
instance. This is easily observed by noting that the number of nodes created to form a
tree in the corresponding instance of $MWTM$ is $2n+3m+1$, and $O(n+m)$ additional
processing is needed in the worst case for every node as its weight to $n+m$ tasks are
all initialized to either $0$ or $1$ resulting in a total time proportional to $O((n+m)^2)$.

\begin{theorem}\label{theo:NP-hard}
$MWTM$ problem is NP-hard.
\end{theorem}

\begin{proof}
It follows easily from Lemma~\ref{lem:iff}, and the fact that transformation is
polynomial in the size of the given E3-SAT instance.
\qed
\end{proof}

MAX-E3-SAT is an optimization problem which generalizes E3-SAT in such a way that
instead of a satisfying assignment, it finds an assignment satisfying the maximum
number of clauses in a given 3-CNF formula. As shown in \cite{H01}, it is NP-hard
to approximate satisfiable MAX-E3-SAT instances to within a factor $7/8+\epsilon$
of the optimal for any $\epsilon \in (0,1]$. It is accordingly noted at this point that we can
slightly modify the illustrated transformation from E3-SAT to $MWTM$ to obtain a
transformation also from MAX-E3-SAT to $MWTM$. First, the weights of assigning
the variable nodes to the corresponding enforcement tasks are set to $m$
(the number of clauses). Then, $m$ additional dummy nodes whose weights
of executing any one of the tasks have all been initialized to zero are introduced as
children directly to the root. It is now easily seen that any given instance of
MAX-E3-SAT denoted by $\Pi_1$ has a solution with value $k^*$ if and only if
the corresponding instance of $MWTM$ denoted by $\Pi_2$ has a solution with
a value of $k^*+mn$. This polynomial time reduction can also be used to
establish that $MWTM$ cannot have a polynomial-time approximation scheme (PTAS).
Otherwise, we could use it to obtain a $7/8+\epsilon$ approximation algorithm
for MAX-E3-SAT, and hence a contradiction. In order to see this, let us assume
that $MWTM$ has a $1-\delta$ approximation where $\delta \in (0,1]$.
For a given instance of MAX-E3-SAT, the corresponding instance of $MWTM$
is first obtained in polynomial time using the transformation just depicted.
Setting $\delta=\frac{1}{mn}(1/8-\epsilon)$, the approximation algorithm
for $MWTM$ is run next on the transformed instance to return
$k+mn \ge (1-\delta)(k^*+mn)$ where $k$ and $k^*$ are the number of
clausal tasks in the approximate and optimal solutions respectively.
We can then write the inequality $(k^*+mn)-(k+mn) \le (k^*+mn)-(1-\delta)(k^*+mn)$.
Arranging the left and the right hand sides, we obtain $k^*-k \le \delta(k^*+mn)$.
For sufficiently large values of $m$ and $n$, the inequality can be rewritten as
$k^*-k \le \delta mnk^*$ which is, in turn, arranged to give $(1-\delta mn)k^* \le k$.
Substituting the value for $\delta$, $(7/8+\epsilon)k^* \le k$ is readily obtained
contradicting the fact that no such approximation is possible unless $P=NP$.
A very trivial result can hence be stated as in the following corollary.

\begin{corollary}\label{theo:NP-hard}
There exists no $1-\epsilon$ approximation algorithm for $MWTM$ problem where $\epsilon \in (0,1]$
unless $P=NP$.
\end{corollary}

\section{ILP Formulation of \textit{MWTM} Problem}\label{sec:ILP}

In an instance of $MWTM$, the number of nodes organized as a tree, $T$, and the number
of tasks are given by $n$ and $m$ respectively. The weight of executing each task $j$ by
a node $i$ is also denoted by $w_{i,j}$ where $i \in [1..n]$  and $j \in [1..m]$. Let $r$
designate the root of this tree, $T$. Let us denote by $ \lambda \subseteq \{1..n\}$
the leaf nodes of $T$. Each unique path from the root $r$ to a leaf node $k \in \lambda$
is represented by a set of nodes on this path which is denoted by $\Pi_k$. Integer Linear
Programming (ILP) formulation of $MWTM$ problem
can thus be given as:

\begin{align}
&\text{maximize } \sum_{i=1}^{n} \sum_{j=1}^{m} w_{i, j}*x_{i, j} \label{eqn:LP}\\
&\text{subject to} \nonumber\\
&\sum_{j=1}^{m}x_{i, j} \le 1, \forall i \in \{1..n\} \label{eqn:LP1} \\
&\sum_{i=1}^{n}x_{i, j} = 1, \forall j \in \{1..m\} \label{eqn:LP2} \\
&\sum_{i \in \Pi_k} \sum_{j=1}^{m} x_{i, j} \le 1, \forall k \in \lambda \label{eqn:LP3} \\
&x_{i, j} \in \{0, 1\}, \forall i \in \{1..n\} \text{ and } \forall j \in \{1..m\} \label{eqn:LP4}
\end{align}

The inequality in (\ref{eqn:LP1}) simply means a node can be assigned to at most one
task. The constraint in (\ref{eqn:LP2}) is used to enforce that every task is executed
by a single node. In order to enforce that on any path leading to a leaf node, at most
one node can be assigned to a task, (\ref{eqn:LP3}) is used. Finally, (\ref{eqn:LP4}) is
there to make sure that decision variables $x_{i, j}$ can take on the integer values $0$
and $1$ only. The given ILP formulation can readily be relaxed to an LP by removing
the last constraint (\ref{eqn:LP4}) which restricts $x_{i, j}$ values to either $0$ or $1$. 

\section{Bottom-Up Assignment Heuristic}\label{sec:bottom-up}

In this section, a heuristic solution is developed in an effort to solve $MWTM$ effectively.
When ILP formulation is relaxed by removing the last constraint (\ref{eqn:LP4}) to obtain
an LP model, $x_{i, j}$ can take on fractional values in the range $[0,1]$. To cope with
these fractional values in order to come up with a feasible integer solution,
\textsc{Bottom-Up-Assignment} (BOA) procedure given in Algorithm~\ref{alg:bottom-up} is used.

Before giving a detailed explanation of BOA, a high level description of the heuristic
can be presented as follows: First, a call is made to obtain a solution
to LP relaxation of ILP formulation of $MWTM$. Then, this possibly fractional
solution is converted to a feasible, partial 0-1 solution where leaf nodes with
greater fractional assignments are favored. The remaining nodes and tasks
that are still not allocated at this current episode, if any, form a smaller
instance of $MWTM$ which is simply handed over to a subsequent iteration.
At this successive iteration, a new call to LP relaxation for the smaller
instance is issued. This process is repeated as long as there are tasks
not assigned yet. The entire heuristic hence works its way via making
leaf-assignments between successive calls to LP.

\SetAlgoSkip{smallskip}
\begin{algorithm}[htb]
\begin{scriptsize}
\SetLine
\KwIn{$T$ is a tree modeling the parent-child relationships among
$n$ nodes rooted at node $r$; $w$ is a 2-dimensional array where
$w_{i, j}$ denotes the weight of assigning node $i$ to task $j$ for all $i \in \{1..n\}$, and $j \in \{1..m\}$.
The number of tasks, given by $m$, satisfies $m > 1$ as $m = 1$ case is trivial to handle;
The weight function is such that $w_{r, j} = 0$ for all $j \in \{1..m\}$ since the root can only be
assigned when there is only one task in the problem instance.}
\KwOut{A feasible assignment $\alpha$ of $m$ tasks to nodes in $T$.}
$\alpha \leftarrow \emptyset$\; \label{boa:initAssignments}
$tasksLeft \leftarrow \{1..m\}$; $nodesLeft \leftarrow \{1..n\}$\label{boa:initTasksNodes}\;
\tcp{a call to LP with substitutions $x_{i, j} \leftarrow 1$  $\forall (i, j) \in \alpha$}
$x_{i, j} \leftarrow LP(T, w, n, m, \alpha)$ $\forall i \in \{1..n\} \text{ and } j \in \{1..m\}$;\label{boa:LP}
\tcp{check for a feasible solution!}
$T' \leftarrow deleteNodes(T, \{i\vert(i, j) \in \alpha\})$\label{boa:initTmpT}; \tcp{delete all nodes assigned}
$\lambda \leftarrow leaves(T')$\label{boa:initLambda}\;
\tcp{leaves with a non-zero assignment are examined in decreasing order of $x_{i, j}$ values}
\While{$(\max x_{i, j} \ne 0$ where $(i \in \lambda \cap nodesLeft)$ and $(j \in tasksLeft)$ can be found$)$
\label{boa:mainWhile}}
{
 	$\alpha \leftarrow \alpha \cup \{(i, j)\}$; \tcp{record this assignment}\label{boa:updateAlpha}
	$\lambda \leftarrow \lambda - \{i\}$\label{boa:updateLambda}\;
	$tasksLeft \leftarrow tasksLeft - \{j\}$; $nodesLeft \leftarrow nodesLeft - \{i\}$\label{boa:updateTasksNodes}\;
	$\Pi_{i} \leftarrow$ set of nodes on the path from $i$ to $r$ in $T'$\label{boa:nodesOnPath}\;
	\tcp{remove all the nodes from $i$ up to $r$ in $T'$ from consideration}
	\lForEach{$k \in (\Pi_{i}-\{i\})$}
	{
		$nodesLeft \leftarrow nodesLeft - \{k\}$\label{boa:hierarchy}\;
	}
	$T' \leftarrow deleteNodes(T', \{i\})$\label{boa:deleteFromTmpT}; \tcp{delete $i$ in $T'$}
}\label{boa:endWhile}
\If{$(tasksLeft \ne \emptyset)$\label{boa:ifLine}}
{
	$T' \leftarrow deleteNodes(T', \lambda)$;\label{boa:deleteRemaining} \tcp{to give ancestors a chance}
	$leavesLeft \leftarrow nodesLeft \cap leaves(T');$ $nodesLeftInT' \leftarrow nodesLeft \cap nodes(T');$
	\label{boa:initLeavesAndNodesLeftInTtmp}\\
	\uIf{$(nodesLeftInT'$ has nodes with $x_{i, j} \ne 0)$ and
	$(\vert leavesLeft \vert \ge \vert tasksLeft \vert)$\label{boa:newCall}}
	{
 		\textbf{go to} step~\ref{boa:initLambda}\label{boa:gotoWhile}\;
 	}
	\Else
	{
		\textbf{go to} step~\ref{boa:LP}\label{boa:gotoLP};
	}
}
\Return{assignment $\alpha$}\label{boa:return}\;
\caption{\textsc{Bottom-Up-Assignment}($T, w, n, m$)}
\label{alg:bottom-up}
\end{scriptsize}
\end{algorithm}

BOA in Algorithm~\ref{alg:bottom-up} assumes that the number of tasks and nodes
are represented by $m$ and $n$ respectively. The number of tasks, $m$, is greater
than $1$ to address only the non-trivial instances of $MWTM$. It is also assumed
that the root of the tree is dummy, i.e., it cannot be assigned to a task as the other
nodes would be rendered useless otherwise. The input to this algorithm are a tree $T$ with
$n$ nodes, and weights $w_{i, j}$ for each node-task pair $(i, j)$ of performing task
$j$ by node $i$.

BOA starts by initializing the set $\alpha$ of assignments at line~\ref{boa:initAssignments}
to be empty. At line~\ref{boa:initTasksNodes}, both sets $tasksLeft$ and $nodesLeft$ used
to keep track of the remaining tasks and the remaining nodes respectively are initialized. The
call to LP, next at line~\ref{boa:LP}, takes as parameters the original $MWTM$ instance
along with the assignments made so far to construct and also solve the LP formulation
given by (\ref{eqn:LP}) through (\ref{eqn:LP3}) of the given $MWTM$ instance with
respect to the set of already made assignments in $\alpha$. If a feasible solution exists, a 2-dimensional
array $x$ of possibly fractional values are returned by this call. The effect of the parameter
$\alpha$ is to set all $x_{i, j}$ values to $1$ in the corresponding LP formulation for all
node-task pairs $(i, j) \in \alpha$. This definitely ensures that neither the ancestors nor
the descendants of already assigned nodes in any feasible solution can have a non-zero
assignment value $x_{i, j}$ associated with them. Line~\ref{boa:initTmpT} deletes all
the nodes in $T$ assigned so far by BOA to obtain a new tree $T'$. This tree $T'$ along
with $nodesLeft$, $tasksLeft$, and the unmodified weight function $w$ actually identify
a residual $MWTM$ instance obtained by reflecting the current assignments in $\alpha$ made so
far into the original instance. This is achieved simply by pruning the assigned nodes, and hence their
descendants from $T$ to obtain $T'$ as well as keeping the set $nodesLeft$ synchronized
in the algorithm by removing the ancestors of these already assigned nodes from it to enforce
the hierarchy constraint. It should be observed at this point that the most recent LP relaxation
formulation at line~\ref{boa:LP} corresponds exactly to this residual $MWTM$ instance as
represented by the current values of the variables in $(T', nodesLeft, tasksLeft, w)$ held
at the time when line~\ref{boa:initTmpT} gets executed. The first parameter to
function $deleteNodes()$ is immutable, and is not modified in the function.
After the set $\lambda$ is populated with a copy
of the leaf nodes in $T'$ at line~\ref{boa:initLambda}, those leaves with a non-zero
assignment in it are examined in the order of non-increasing $x_{i, j}$ values
in the while-loop between lines \ref{boa:mainWhile} through \ref{boa:endWhile}.
The loop iterates as long as the maximum value assignment with a non-zero
$x_{i, j}$ between the leaf nodes and the tasks not assigned yet can be found.
Among the leaves in $\lambda$ which have not been assigned to a task yet, only the
ones not removed due to the hierarchy constraint are considered eligible as reflected
by the expression $(i \in \lambda \cap nodesLeft)$ where $nodesLeft$ keeps track of
the remaining nodes in the original tree $T$ which have yet been neither assigned nor
left in a non-assignable state as a result of the hierarchy constraint. Existence
of such an $x_{i, j}$ value requires that an assignment between node $i$ and task
$j$ gets recorded as illustrated at line~\ref{boa:updateAlpha}. This line amounts
effectively to setting $x_{i, j}$ to $1$. The following two lines \ref{boa:updateLambda}
and \ref{boa:updateTasksNodes} updates accordingly the set of leaves not considered
yet and the sets of remaining tasks and nodes after the assignment just made. Since this
recent assignment of node $i$ also necessitates that the nodes on the path from node $i$
up to the root $r$ are removed from any further consideration for a possible assignment, such
nodes computed at line~\ref{boa:nodesOnPath} are accordingly deleted from $nodesLeft$
at line~\ref{boa:hierarchy}, and get the right treatment. As the final statement
at line~\ref{boa:deleteFromTmpT} in the body of the while-loop, node $i$ is pruned
from $T'$ by the respective call. The thread of control is transferred to line~\ref{boa:ifLine}
to test whether any tasks have been left not assigned, as soon as it breaks out of the loop.
If all tasks have already been assigned, the set of assignments constructed so far is
returned at line~\ref{boa:return} as the solution. Otherwise, the remaining leaf nodes
are first deleted from $T'$ at line~\ref{boa:deleteRemaining} to give their ancestors a
chance before a new call to LP is made. Then at line~\ref{boa:initLeavesAndNodesLeftInTtmp},
a set of leaves in $T'$ that are also in $nodesLeft$ denoted by $leavesLeft$, and a set of all
the nodes in $T'$ that are also in $nodesLeft$ represented by $nodesLeftInT'$ are computed.
Finally, at line~\ref{boa:newCall}, a conditional check consisting of the conjunction of
two expressions is performed. The former expression evaluates to true if the nodes
in $T'$ that can still be used for further assignments have non-zero $x_{i, j}$ values
with $j \in tasksLeft$. The latter expression called the \emph{feasibility invariant} is
maintained throughout the entire execution of the algorithm. It basically ensures
that the number of the leaf nodes still assignable are always greater than the number of
the remaining tasks. If both expressions evaluate to true, execution continues by
setting $\lambda$ to the leaf nodes in the updated $T'$ at line~\ref{boa:initLambda}
to get ready for the subsequent execution of the while-loop once more, and
otherwise a jump to line~\ref{boa:LP} occurs where a new invocation to LP occurs.
All the deletions performed at line~\ref{boa:deleteRemaining} in $T'$ are effectively
rolled back at line~\ref{boa:initTmpT}.

Both $deleteNodes()$ and $leaves()$ which are based on post-order traversal
run in time proportional to the number of nodes in the tree they operate on.
An implementation making an efficient evaluation at the start of every iteration
of the while-loop possible employs max-heaps one for every task not assigned
yet whose roots are also organized as a max-heap. Overall running time complexity
of the heuristic is, however, dominated by the calls to LP at line~\ref{boa:LP}. As
after each call, if a feasible solution exists, BOA assigns at least one task before
the next call to LP, the total number of LP calls made is equal to the number of
tasks, $m$, in the worst case. Since LP lends itself to polynomial solutions
\cite{Khachiyan79, Karmarkar84}, BOA is easily demonstrated to be also
polynomial in its worst case running time. The overhead originating from
the repetitive nature of the heuristic is discussed also in the next section, and it is
shown through experiments that the actual observed value for the number of times
the call at line~\ref{boa:LP} to LP gets executed is almost constant on the average.

If there is a 0-1 assignment to ILP formulation of a given $MWTM$ instance,
its LP relaxation has certainly a fractional assignment with total weight at least
that of ILP. In such a case, this fractional assignment can always be converted
to a feasible 0-1 assignment by BOA in Algorithm~\ref{alg:bottom-up}. In an
effort to prove this, a series of lemmas will be presented and some observations
regarding the algorithm will be made.

A trivial observation could be made at this point by simply noting that
a condition which ensures that the number of leaf nodes is greater
than or equal to the number of tasks in a given instance of $MWTM$
is both necessary and sufficient for the existence of a solution.

\begin{lemma}\label{lem:MWTMfeasible}
A given instance of $MWTM$ represented by $(T, w, m, n)$ where $T$
is a tree, and $w(i, j)$ is the weight of assigning node $i$ in $T$ to task $j$
for all combinations of $i \in \{1..n\}$ and $j \in \{1..m\}$ has a solution
if and only if $\vert \lambda \vert \ge m$ where $\lambda$ denotes the
set of leaf nodes in $T$.
\end{lemma}
\begin{proof}
Let us prove the sufficiency part first. If a given $MWTM$ instance $(T, w, m, n)$
has a solution, then there exists an assignment of $m$ nodes in $T$ to $m$ tasks.
The hierarchy constraint in the definition of $MWTM$ problem requires, in turn, that
no two among these $m$ nodes has a parent-child relationship, and they are,
hence, on $m$ mutually independent paths (see Definition~\ref{def:independentPath}).
Therefore, the number of leaves in $T$ denoted by $\vert \lambda \vert$ cannot be
less than the number of available independent paths from these $m$ nodes to the root.

In order to prove the necessity part, we proceed as follows: As there are as
many as $\vert \lambda \vert \ge m$ leaf nodes, any subset of $m$ leaves
out of $\lambda$ can be freely picked, and assigned to available tasks in a
random order. Since each node that gets picked is on an independent path
ensuring that the hierarchy constraint is not violated, a feasible solution is
hence obtained.
\qed
\end{proof}

\begin{definition}\label{def:effectiveLeaf}
In a feasible solution to LP relaxation of a given $MWTM$ instance,
a node $i$ in tree $T$ associated with at least one non-zero $x_{i, j}$, and yet, not having any
such descendants in $T$ is defined to be an \textit{effective leaf} with
respect to the corresponding LP relaxation solution. The set of all
such nodes is termed \textit{effective leaves}.
\end{definition}

In the light of this definition, the following lemma can now be stated
regarding an LP relaxation formulation corresponding to a given
$MWTM$ instance.

\begin{lemma}\label{lem:LPfeasible}
If LP relaxation to a given $MWTM$ instance has a solution, then
the number of effective leaf nodes in the corresponding LP relaxation
is greater than or equal to the number of tasks in the given problem instance.
\end{lemma}
\begin{proof}
If a given $MWTM$ instance's LP relaxation has a solution, then the constraints
(\ref{eqn:LP1}) through (\ref{eqn:LP3}) must hold. Therefore, we obtain by
summing Equation~(\ref{eqn:LP2}) over all possible $j$ values:
\begin{equation}\label{eqn:totalAssignment}
\sum_{j=1}^{m}\sum_{i=1}^{n} x_{i, j}=m
\end{equation}
Let $\lambda_{e}$ denote the set of effective leaves in $T$ with respect to the
particular LP relaxation solution. Since no nodes other than those in $\lambda_{e}$
and their ancestors can have a non-zero $x_{i, j}$ value associated with them,
we next sum Inequality~(\ref{eqn:LP3}) over all the effective leaf nodes to obtain:
\begin{equation}\label{eqn:totalPathAssignment}
\sum_{k \in \lambda_{e}}\sum_{i \in \Pi_k} \sum_{j=1}^{m} x_{i, j} \le \vert\lambda_{e}\vert
\end{equation}
As the sum of individual $x_{i, j}$ values in (\ref{eqn:totalAssignment}) is less than or equal to
the sum in (\ref{eqn:totalPathAssignment}) over all paths leading to effective leaf
nodes, we conclude:
\begin{equation}\label{eqn:sufficiency}
m \le \vert\lambda_{e}\vert
\end{equation}
\qed
\end{proof}

We can now establish the following lemma by noting that the number of
effective leaves with respect to the corresponding LP relaxation solution
of a given $MWTM$ instance actually forms a lower bound for the number
of leaf nodes in $T$.

\begin{lemma}\label{lem:LPlambda}
The corresponding LP relaxation of a given $MWTM$ instance, $(T, w, m, n)$,
has a solution if and only if $\vert \lambda \vert \ge m$ where $\lambda$
denotes the set of leaf nodes in $T$.
\end{lemma}
\begin{proof}
As to the sufficiency; if LP relaxation has a solution, then, by Lemma~\ref{lem:LPfeasible},
$\vert \lambda_{e} \vert \ge m$ where $\lambda_{e}$ is the set of effective leaves. As
it is known that $\vert \lambda \vert \ge\vert \lambda_{e} \vert$,
$\vert \lambda \vert \ge m$ follows easily.

In order to prove the necessity, on the other hand, we observe by
Lemma~\ref{lem:MWTMfeasible} that if $\vert \lambda \vert \ge m$, then
the given $MWTM$ instance has a solution. This latter result definitely
implies the existence of a solution to the corresponding LP relaxation
formulation.
\qed
\end{proof}

\begin{definition}\label{def:iteration}
An execution of BOA between successive calls to LP at line~\ref{boa:LP} is
called an \textit{iteration}.
\end{definition}

\begin{theorem}\label{theo:BOAfeasible}
BOA heuristic in Algorithm~\ref{alg:bottom-up} returns a feasible solution whenever there exists one.
\end{theorem}
\begin{proof}
The algorithm will keep repeatedly performing iterations until all tasks in $tasksLeft$ are exhausted,
and finally an assignment is obtained. Every single iteration of BOA is launched at line~\ref{boa:LP}
in an attempt to discover an assignment for a smaller residual $MWTM$ instance which is identified
by the values that the tree $T'$, $nodesLeft$, $tasksLeft$, and the unmodified weight function $w$
have right after the statement at line~\ref{boa:initTmpT} gets executed. The most recent LP
relaxation formulation at line~\ref{boa:LP} corresponds exactly to this instance whose full recognition
is achieved interestingly enough later at the next statement.

It is known by Lemma~\ref{lem:LPlambda} that when the number of leaf nodes in $nodesLeft$
(given by $leaves(T') \cap nodesLeft$ as would be computed at line~\ref{boa:initLambda}) is
greater than or equal to the number of remaining tasks in $tasksLeft$ in the residual $MWTM$
instance at the start of an iteration $i$ before a call to LP, there must exist a feasible solution
to the corresponding LP relaxation formulation at line~\ref{boa:LP}. The existence of a feasible
solution to the corresponding LP relaxation, in turn, implies by Lemma~\ref{lem:LPfeasible} that
the number of effective leaf nodes in $T'$ computed at line~\ref{boa:initTmpT} is greater than or
equal to the number of tasks in $tasksLeft$. Therefore, at least one assignment between an
effective leaf and an available task will be performed in the while-loop between lines~\ref{boa:mainWhile}
through \ref{boa:endWhile} in every iteration of the algorithm, and BOA will eventually terminate.

An additional observation can be made by noting that the feasibility variant cannot be violated
so long as the while-loop iterates since it holds at the start, and the only type of modification
allowed in the body of the loop is the assignment of an effective leaf to an available task. Such
an assignment, however, removes exactly one leaf node and one task from consideration
ensuring that the feasibility invariant is still maintained.

Once the control breaks out of the while-loop, either a feasible solution by BOA is returned
if there are no more tasks left, or otherwise all the useless leaf nodes which survived the
previous while-loop are deleted at line~\ref{boa:deleteRemaining}. These leaf node
deletions are the only deletions that can possibly violate the feasibility invariant. Hence,
once such a violation is detected at line~\ref{boa:newCall}, a jump at line~\ref{boa:gotoLP}
initiates the next iteration where all such deletions are effectively rolled back by reconstructing
$T'$ from scratch. In case there are no such violations, control goes once more to the while-loop.
Consequently, the feasibility invariant is maintained from one iteration to the next
throughout the entire execution of the algorithm.

BOA will then always find a feasible solution as long as the feasibility constraint holds at the start
of the first iteration. But this is already guaranteed by Lemma~\ref{lem:MWTMfeasible}, hence
completing the proof.
\qed
\end{proof}

It is clearly not easy, if not impossible, to generate a feasible 0-1 assignment at once using
only possibly fractional non-zero assignments obtained from the corresponding LP relaxation
solution for all $MWTM$ instances. The difficulty stems from the fact that the distribution of
fractional assignment values returned by the corresponding LP relaxation solution may not
easily lend itself to an integer valued assignment for all tasks without violating the hierarchy
constraint. Therefore, as it is done in BOA, LP might need to be called iteratively in order to
cover the tasks that have not been already assigned in the previous iterations. In an attempt
to reduce the number of iterations, however, as many fractional assignment values as allowed
by the feasibility invariant is checked at each iteration in BOA as to their eligibility to contribute
to a feasible solution. Moreover, some zero valued assignments in previous iterations may come
out non-zero in subsequent iterations leading to solutions with smaller total weights. Thus, we
prefer to use the earliest LP results with non-zero assignments as much as possible.

An example is provided below for a better understanding of how BOA operates.

\begin{example}
An $MWTM$ instance is depicted in Figure~\ref{fig:sampleTreeOptILP}. The corresponding
tree representing the hierarchical structure of an organization has $6$ nodes numbered in
level order as shown in the figure where the root node is denoted by $1$. It is assumed in
this particular example that the organization has $3$ tasks to be executed not explicitly
shown in the figure. The weights of executing these tasks, namely $t_1$, $t_2$, and $t_3$
are given in this order as a triple inside each node. In this example, the corresponding ILP
formulation will produce the optimal solution with the following assignments highlighted with
the corresponding weight values in red in Figure~\ref{fig:sampleTreeOptILP}:
\begin{itemize}
\item{Task $t_1$ is assigned to node $4$ with weight $6$,}
\item{Task $t_2$ is assigned to node $5$ with weight $4$,}
\item{Task $t_3$ is assigned to node $3$ with weight $8$.}
\end{itemize}

\begin{figure}[h]
	\centering
	\includegraphics[width=2in]{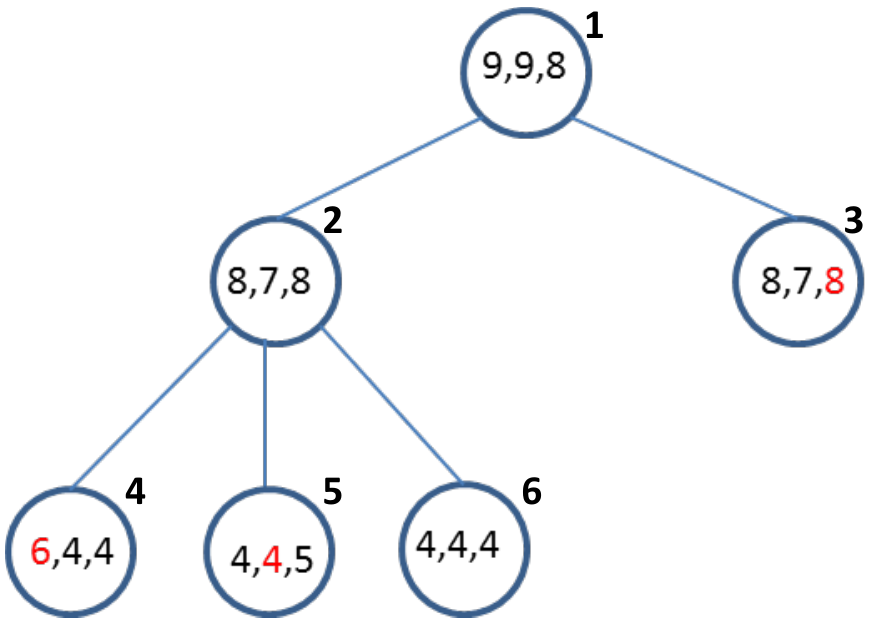}
	\caption{A sample tree structure with $6$ nodes, and $3$ tasks not explicitly
shown. Weights of executing each task $t_1$, $t_2$, and $t_3$ are given in this order
inside each node as triples. Red values correspond to node-task assignments obtained
from ILP solution.}
	\label{fig:sampleTreeOptILP}
\end{figure}

The next figure, Figure~\ref{fig:sampleTreeFirstLP}, presents a solution obtained by the
corresponding LP relaxation on the same problem instance (assignments are shown in green).
The solution is as follows:
\begin{itemize}
\item{Task $t_1$ is assigned to nodes $3$ and $4$ both with the same fractional value
$0.5$, contributing to the total weight by $7 (=8/2+6/2)$,}
\item{Task $t_2$ is assigned to nodes $5$ and $6$ both with the same value $0.5$ again,
contributing to the total weight by $4 (=4/2+4/2)$,}
\item{Finally, Task $t_3$ is assigned to nodes $2$ and $3$ both with the same value $0.5$,
causing this time an increase of $8 (=8/2+8/2)$ in the total weight.}
\end{itemize}

\begin{figure}[h]
	\centering
	\includegraphics[width=2in]{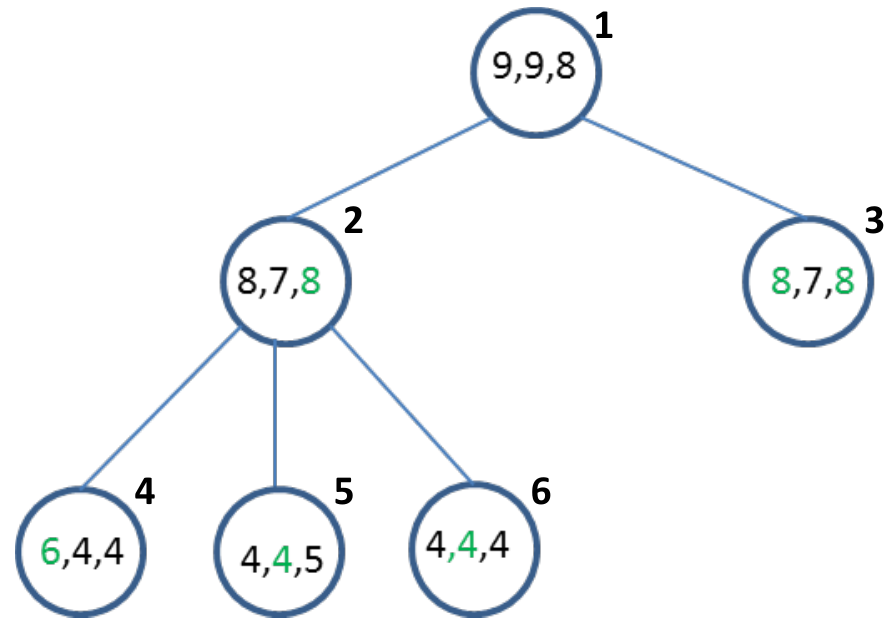}
	\caption{Green values correspond to fractional node-task assignments
obtained from the corresponding LP relaxation solution where task $t_1$ is assigned to both
nodes $3$ and $4$, task $t_2$ is assigned to nodes $5$ and $6$, and finally task $t_3$ is
assigned to nodes $2$ and $3$ all with the same value $\frac{1}{2}$.}
	\label{fig:sampleTreeFirstLP}
\end{figure}

Since LP is allowed to make fractional assignments, the weight $19$ of the solution
achieved by LP is even higher than the optimal $18$ found by ILP. The direct
application of LP unfortunately cannot produce an integer assignment for the
given $MWTM$ instance. BOA in Algorithm~\ref{alg:bottom-up}, however, will
work its way to a feasible solution as follows on this example:
\begin{itemize}
\item{After a call to LP is made at line~\ref{boa:LP} in the first iteration, potentially
fractional assignment values with non-zero $x_{i, j}$ will be processed from the
largest to the smallest for the leaf nodes of the tree. In this example, all the
assignment values happen to be the same, namely $0.5$. Such leaves may,
therefore, be processed in any order. Although different heuristics may also be
developed for breaking ties such as considering the depths of nodes or favoring
nodes with higher $w_{i, j}$ values, we assume for the sake of this example that
the assignments with the same value are processed in increasing order of node
identifiers and then in increasing order of task numbers. As a result, first, task $t_1$
is assigned to node $3$ with weight $8$. Then, task $t_2$ gets assigned to node $5$
with weight $4$. This assignments in the first iteration are shown in green as depicted
in Figure~\ref{fig:sampleTreeBottomUpAssignment}. At this point, there is obviously
no leaf node left with a non-zero assignment that can be used to make any further
assignments, leaving task $t_3$ hence unassigned. It should noted that while these
assignments are made, all the nodes violating the hierarchy constraint are also
removed from consideration. This is evidently reflected by leaving only the nodes
$4$ and $6$ in $nodesLeft$.}
\item{Once it is realized that no more assignments are possible, the remaining leaves,
namely $4$ and $6$, are deleted at line~\ref{boa:deleteRemaining} from the tree $T'$
leaving only the nodes $1$ and $2$ in it. As there are no nodes in the tree that are also in
$nodesLeft$, a jump to line~\ref{boa:LP} initiates the second iteration of the algorithm.}
\item{With $tasksLeft=\{t_3\}$ and $nodesLeft=\{4, 6\}$ at the start of the second
iteration, the only remaining task is $t_3$, and the remaining nodes that are eligible
for assignments are $4$ and $6$. Now a call is made to LP formulated with the
assignments made in the first iteration in mind. This formulation corresponds exactly
to an $MWTM$ instance where the tree denoted by $T'$ is obtained at
line~\ref{boa:initTmpT} by pruning nodes $3$ and $5$ from the original tree
denoted by $T$, and the set of eligible nodes and target tasks to be matched
are as dictated by the values of $nodesLeft$ and $tasksLeft$ at the moment.
The algorithm, hence, terminates by assigning the only remaining task $t_3$
to either node $4$ or node $6$ with weight $4$. Figure~\ref{fig:sampleTreeBottomUpAssignment}
shows this assignment in the second iteration in orange. The total weight
achieved by BOA is hence $8+4+4=16$, which is slightly less than the
optimal ILP solution.}
\end{itemize}

\begin{figure}[h]
	\centering
	\includegraphics[width=2in]{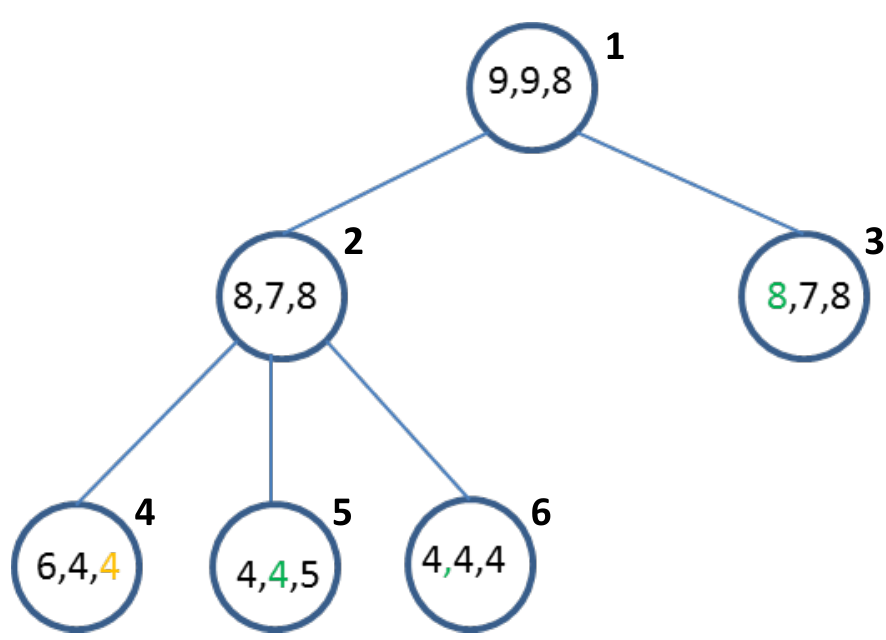}
	\caption{The assignments obtained by two LP calls. The first call generates
assignments for the tasks $t_1$ and $t_2$ (green), and the second call generates
the assignment for the task $t_3$ (orange).}
	\label{fig:sampleTreeBottomUpAssignment}
\end{figure}
\qed
\end{example}

Another example is presented now to demonstrate the feasibility invariant
at line~\ref{boa:newCall} of BOA in Algorithm~\ref{alg:bottom-up}.

\begin{example}
Figure~\ref{fig:needForInvariant} is an example to an $MWTM$ instance where deletions by BOA
at line~\ref{boa:deleteRemaining} of many leaves rooted at the same node renders this parent
as an effective leaf, and the subsequent assignment of this parent to a new task runs the risk of
an unanticipated decrease in the number of leafs leading to unfeasibility.

\begin{figure}[h]
	\centering
	\includegraphics[width=2.5in]{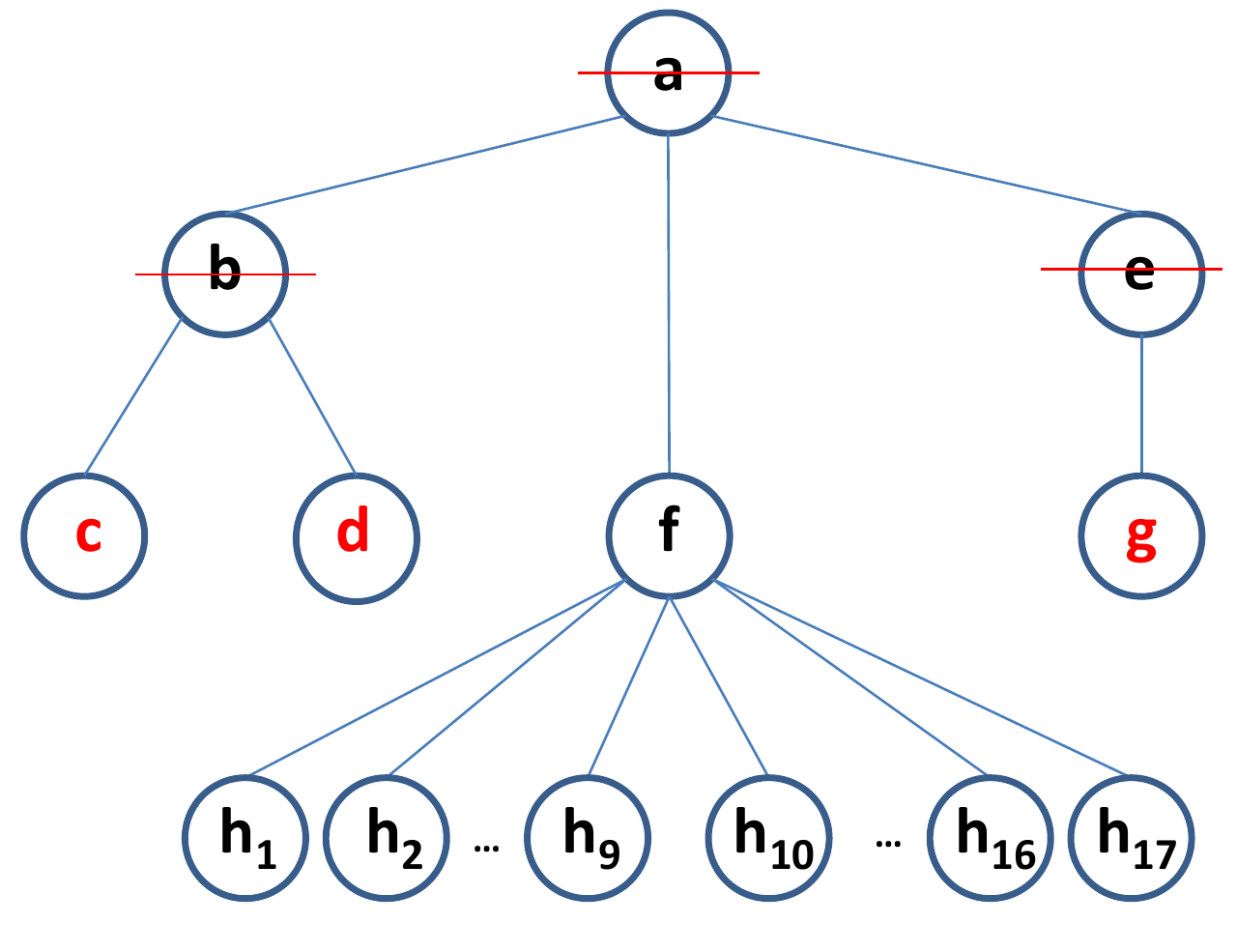}
	\caption{An instance illustrating the need for the invariant
at line~\ref{boa:newCall} of the algorithm.}
	\label{fig:needForInvariant}
\end{figure}

\begin{figure}[h]
	\centering
	\includegraphics[width=4in]{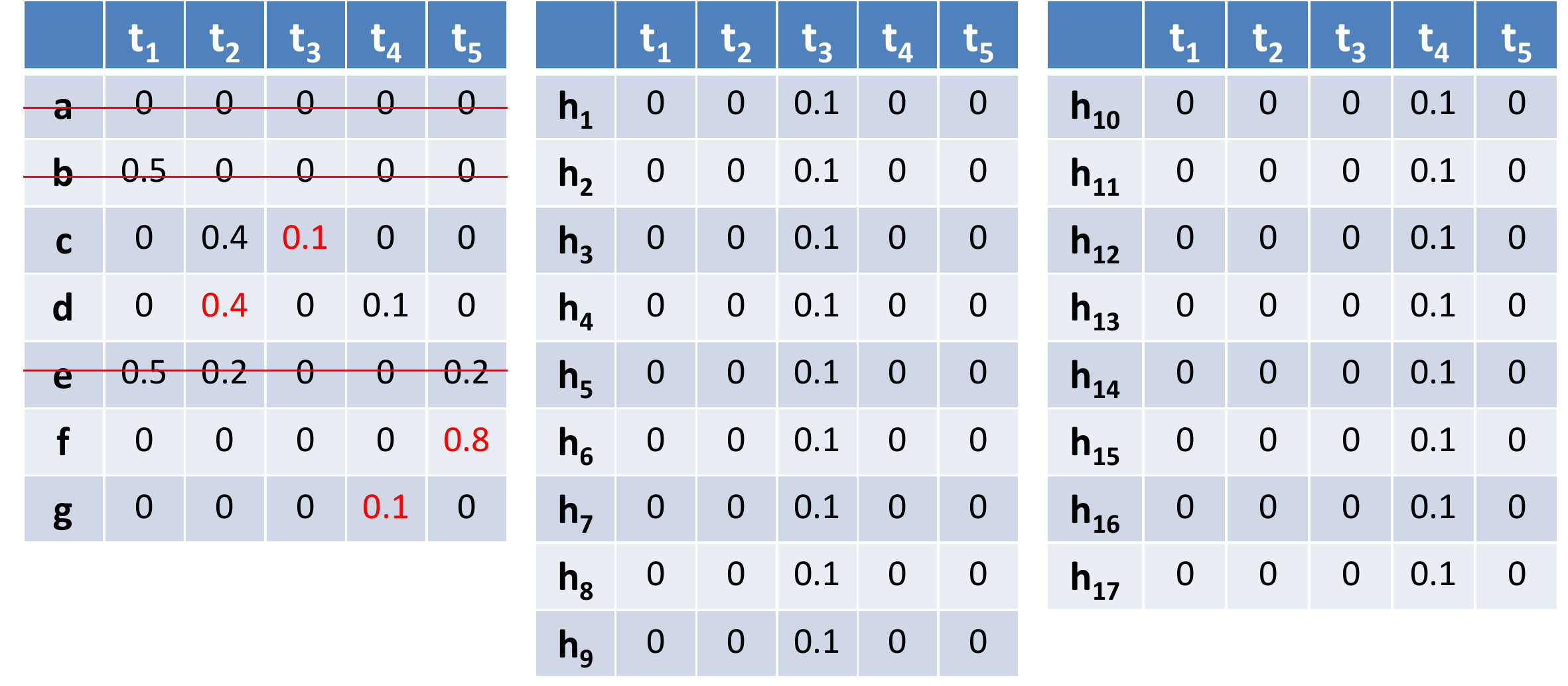}
	\caption{Assignments obtained by the first call to LP for the same instance in Figure~\ref{fig:needForInvariant}.}
	\label{fig:needForInvariantCost}
\end{figure}

Let us assume that there does exist a weight function $w_{i, j}$ such that $x_{i, j}$ values obtained
by an LP call are as given in Figure~\ref{fig:needForInvariantCost}. Then, the sum of assignments
to each of the five tasks is $1$, and each path from a leaf node to the root denoted by $a$ has total
weight less than or equal to $1$. If the algorithm is run, after the first call to LP, leaf node $d$ with the
maximum assignment value $0.4$ (breaking ties arbitrarily) gets assigned to task $t_2$. This assignment
makes the next highest assignment value $0.4$ for node $c$ unusable leaving us with the only option
of $0.1$ as to the next largest assignment value. Assuming that node $c$ is now picked to be assigned
to task $t_3$, followed by the same valued assignment of node $g$ to task $t_4$, all leafs $h_1$ through
$h_{17}$ are also rendered useless. As a result, they are all deleted at line~\ref{boa:deleteRemaining}.
This, in turn, would give way through a jump at line~\ref{boa:gotoWhile} to assigning node $f$ to
task $t_5$ within the body of while-loop, if it were not for the invariant at line~\ref{boa:newCall} in
the algorithm. Such an assignment clearly would have left no nodes that can be assigned to task $t_1$.

The feasibility invariant at line~\ref{boa:newCall} of the algorithm ensures that a sufficient
number of leaves to a possible next iteration is always maintained.
\qed
\end{example}

\section{Experiments}\label{sec:experiments}

In order to measure the performance of LP-relaxation based heuristic
BOA in Algorithm~\ref{alg:bottom-up}, several experiments have been
performed for varying problem parameters. The parameters employed,
and their values are as follows:
\begin{enumerate}
\item{\textit{\#Nodes}: It represents the number of nodes in the tree in a given $MWTM$
instance. In order to generate a variety of tree sizes, the following values are employed in
the experiments: $16$ (\textit{small tree}), $32$, $64$, and $128$ (\textit{large tree}).}
\item{\textit{Average Degree}: This parameter is defined to be the average degree of
a node in the tree in a given instance of $MWTM$. It is tuned throughout the experiments
to control the type of trees generated in a scale ranging from \textit{deep} to \textit{shallow}
for fixed values of \textit{\#Nodes} parameter.
The values used in the experiments are $1.5$ (deep tree), $2.0$, and $2.5$ (shallow tree).}
\item{${}^{\#Tasks}/{}_{\#Nodes}$: It is defined to be the ratio of the number of the tasks
to the number of the nodes in the tree associated with a given $MWTM$ instance. This parameter
is used to generate a range of $MWTM$ instances changing from those with a very few tasks
called \textit{sparse} to those with a large number of tasks called \textit{dense} in proportion
to the tree size. The values used are $0.125$ (sparse), $0.25$, and $0.5$ (dense). As this
ratio increases, the flexibility to use non-leaf nodes for assignments decreases.}
\item{\textit{Weight Distribution}: The weight of assigning a node to a task has a value
chosen from the range $[1..\frac{\#Nodes}{2}]$. The following $3$ weight distributions
are used: i) the weights are increasing from the root to the leaves, ii) the weights are
decreasing from the root to the leaves, and iii) the weights are assigned randomly without
regard to the respective depths of the nodes.}
\end{enumerate}

For each combination of these four parameters, a total of $4*3*3*3 = 108$ different
test cases are formed. For each test case, $20$ instances of the problem are then randomly
generated, and their averages are taken in the experiments. We record the total number
of LP calls made at line~\ref{boa:LP} in BOA for every instance.
Corresponding to each instance, both the execution time and the solution obtained
are also recorded once for the corresponding ILP formulation which gives the optimal
solution, and once for BOA expected to return a suboptimal solution.

All the tests were run on a machine with a 4 GB of RAM and an
Intel Core 2 Duo T9550 2.66 Ghz mobile processor. Microsoft Solver Foundation 3.0 was
employed as LP/ILP solver library, and the code was developed in $C\# 5.0$.

The results of the experiments are presented through a series of seven
tables in this section. These tables all share a common structure. As the
topmost two rows are used to set the values for the parameters
\textit{Average Degree} and ${}^{\#Tasks}/{}_{\#Nodes}$,
the leftmost two columns display the values for the parameters
\textit{Weight Distribution} and \textit{\#Nodes}. The last six tables,
on the other hand, can be logically grouped into three each with two
tables. While the first table in a group presents a comparison between
the execution times of ILP and BOA, the second evaluates the quality
of the solutions by BOA against the optimal. These three groups correspond
to the three distinct values that the \textit{Average Degree} parameter can
take on, namely $2.5$, $2.0$, $1.5$, and are presented in this order.
Of the four parameters only one, namely the \textit{Average Degree},
is fixed, and the average results are given for all combinations of the other
three parameters in these groups of tables. Finally, an additional row
labeled \textit{Method} is inserted as the third from the top to allow us
to specify either ILP or BOA in these tables. It should be noted that the
cells at the same position in both tables in the same group correspond
to the exact same combination of parameter values.

The colors yellow and green are used consistently to highlight the cells containing \textit{NaN}
and $\infty$ respectively in all the tables. The cells in yellow marked with \textit{NaN} in a table
mean that there exists no feasible solution. For some combinations of parameters no feasible
solution was possible. Especially when the instances get dense, and the trees associated with
them become deep, as would be expected, it becomes more difficult to find a feasible solution
satisfying the hierarchy constraint. Such configurations are characterized with high
${}^{\#Tasks}/{}_{\#Nodes}$ values, and with the low values of the \textit{Average Degree}
parameter. The results in the tables to follow confirm this expectation. All such cases leading to
infeasibility are shown in yellow. Moreover, when the weight distribution is such that it is
decreasing from the root to the leaves, finding an optimal solution becomes even more difficult
using ILP. Under these circumstances, the execution time for ILP grows very quickly after the
number of nodes become larger than $16$. We do not include these extremely large execution
times in the tables, and indeed we have canceled those solutions without finding the optimal
values. All such cells are displayed in green marked with an $\infty$ symbol. The existence of
feasible solutions by BOA in the corresponding cells, on the other hand, is an evidence for the
existence of the optimal solutions for those cases as well. In order to verify, therefore, the
quality of a solution by BOA in these situations, we make use of the corresponding possibly
fractional LP relaxation solution as a potential upper bound. A quick inspection of the relevant
cells reveals that the difference is very small even in these cases which definitely guarantees
an even smaller distance to the actual optimal. It is hence suspected that BOA might even
have achieved it.

\begin{figure}[htb]
	\centering
	\includegraphics[width=4.2in]{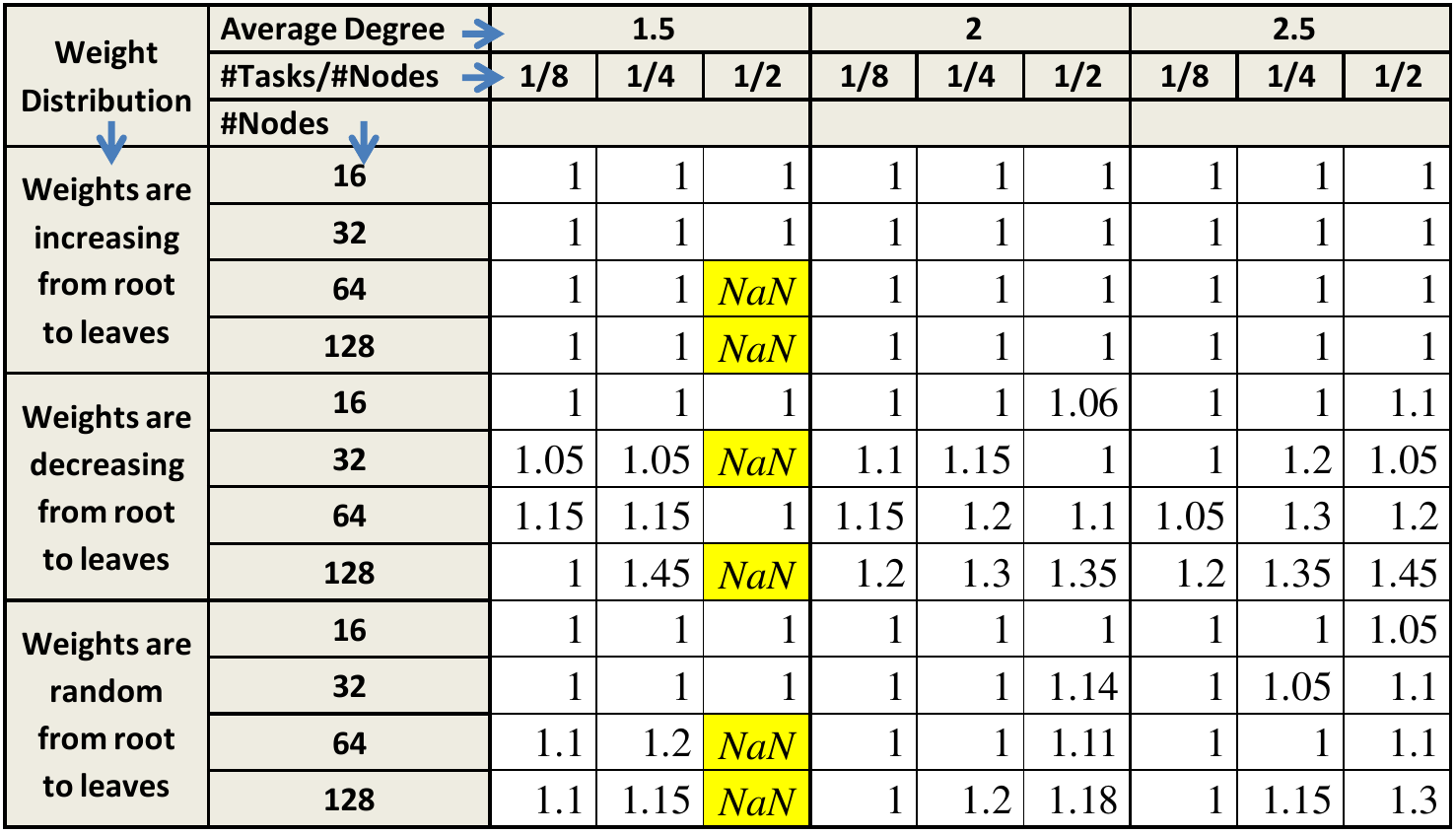}
	\caption{The average number of times LP is called at line~\ref{boa:LP} in BOA
in Algorithm~\ref{alg:bottom-up} for all test cases. The cells in yellow are marked with the
symbol \textit{NaN} to mean that there exists no feasible solution.}
	\label{fig:expNoOfLPCalls}
\end{figure}

The table in Figure~\ref{fig:expNoOfLPCalls} displays the average number of LP
invocations performed at line~\ref{boa:LP} in BOA in Algorithm~\ref{alg:bottom-up}
for each of $108$ different test cases. As the table clearly reflects, the number of
times the call to the corresponding LP relaxation gets executed is very close to $1$.
The cells marked with \textit{NaN} all correspond to the test cases for which no feasible
solutions exist as explained above.

The two tables in Figure~\ref{fig:expTime_2.5} and Figure~\ref{fig:expGoal_2.5}
display the execution times, and the solutions respectively when the parameter
representing the average degree of a node in the tree is set to $2.5$ which
corresponds to shallow trees. There are only $3$ out of $36$ test cases where
BOA is slightly slower in Figure~\ref{fig:expTime_2.5}. These correspond to the
test cases where:
i) ${}^{\#Tasks}/{}_{\#Nodes}$ {\small = 1/8}, {\small \textit{Weight Distribution} = random},
{\small \textit{\#Nodes} = 128},
ii) ${}^{\#Tasks}/{}_{\#Nodes}$ {\small = 1/4}, {\small \textit{Weight Distribution} = increasing},
{\small \textit{\#Nodes} = 64}, and
iii) ${}^{\#Tasks}/{}_{\#Nodes}$ {\small = 1/4}, {\small \textit{Weight Distribution} = random},
{\small \textit{\#Nodes} = 64}.
BOA, on the other hand, achieves optimal or almost optimal solutions as seen in
Figure~\ref{fig:expGoal_2.5} for these test cases. Also an examination of the
cells corresponding to these test cases in the table in Figure~\ref{fig:expNoOfLPCalls}
reveals that they all have the value one.

\begin{figure}[htb]
	\centering
	\includegraphics[width=4.4in]{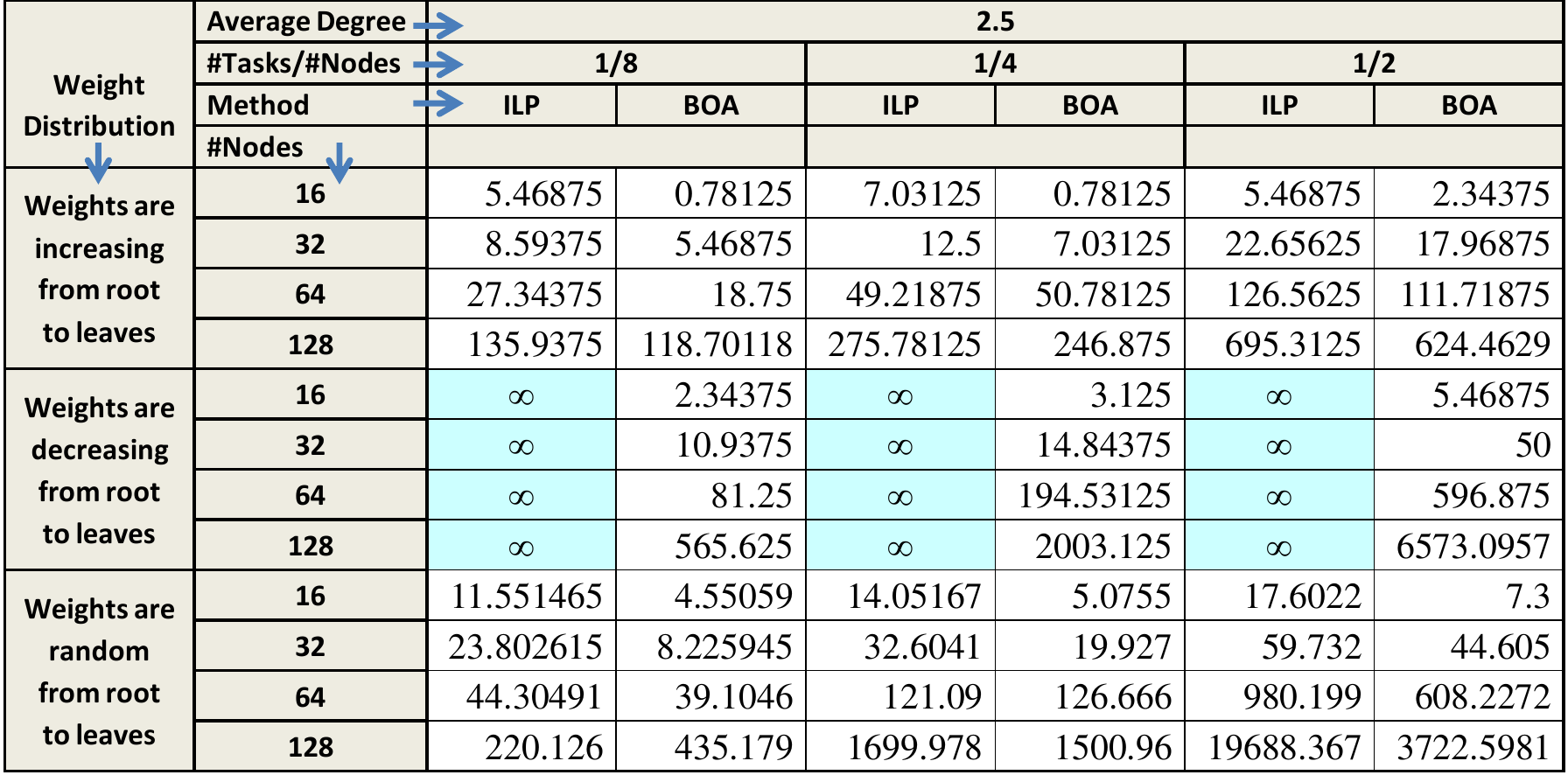}
	\caption{The execution times when the average degree of a tree node
parameter is set to $2.5$ corresponding to shallow trees. The symbol $\infty$
in a blue cell indicates a very large value.}
	\label{fig:expTime_2.5}
\end{figure}

\begin{figure}[htb]
	\centering
	\includegraphics[width=4.4in]{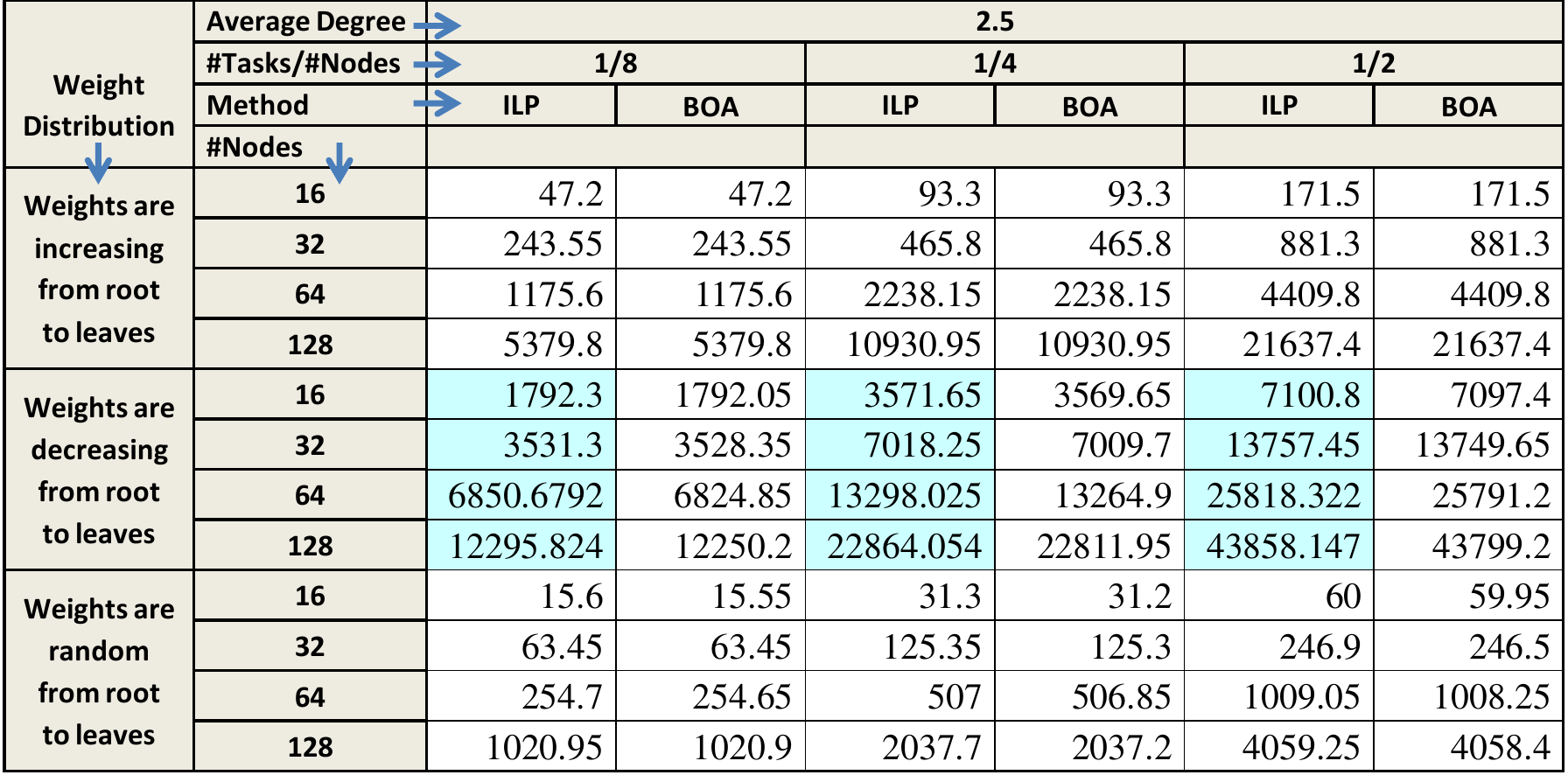}
	\caption{The solutions obtained when the average degree of a tree node
parameter is set to $2.5$. The values in the blue cells are the estimated upper bounds
obtained by the corresponding possibly fractional LP relaxation solutions.}
	\label{fig:expGoal_2.5}
\end{figure}

These execution time anomalies observed to
occur when BOA finds an almost optimal solution in only one iteration can therefore
be explained by the overhead introduced by BOA. When BOA obtains an almost
optimal solution with a single LP call, it would be natural to also expect ILP itself
to discover the optimal integer assignments quickly. As BOA has some additional
computations, its running time for such cases would be slightly more than
that of ILP.

Even when it takes forever to compute the optimal by ILP, the values in the
corresponding blue cells in Figure~\ref{fig:expTime_2.5} are all available for
BOA as an indication of its running time performance. In terms of solution
quality, BOA always achieves optimal solutions when \textit{Weight Distribution}
is such that it is increasing from the root to the leaves. Otherwise, the solutions
obtained as shown in Figure~\ref{fig:expGoal_2.5} are so close to the corresponding
optimal values that it is easily seen to perform within $1\%$ of even the upper bounds
obtained via the corresponding LP relaxation solution.

\begin{figure}[htb]
	\centering
	\includegraphics[width=4.4in]{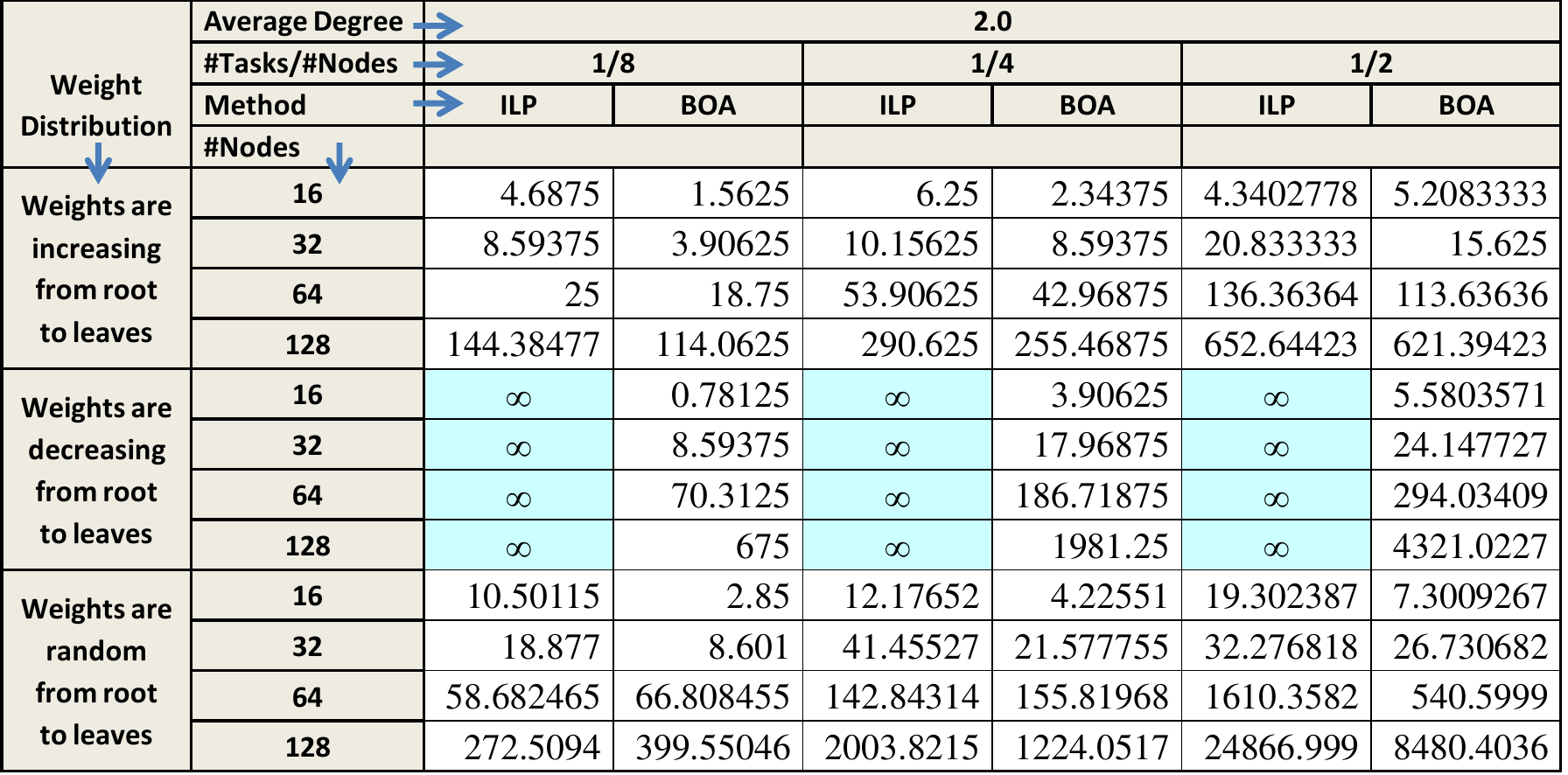}
	\caption{The execution times when the average degree of a tree node
parameter is set to $2.0$. The symbol $\infty$ in a blue cell indicates a very large value.}
	\label{fig:expTime_2.0}
\end{figure}

\begin{figure}[htb]
	\centering
	\includegraphics[width=4.4in]{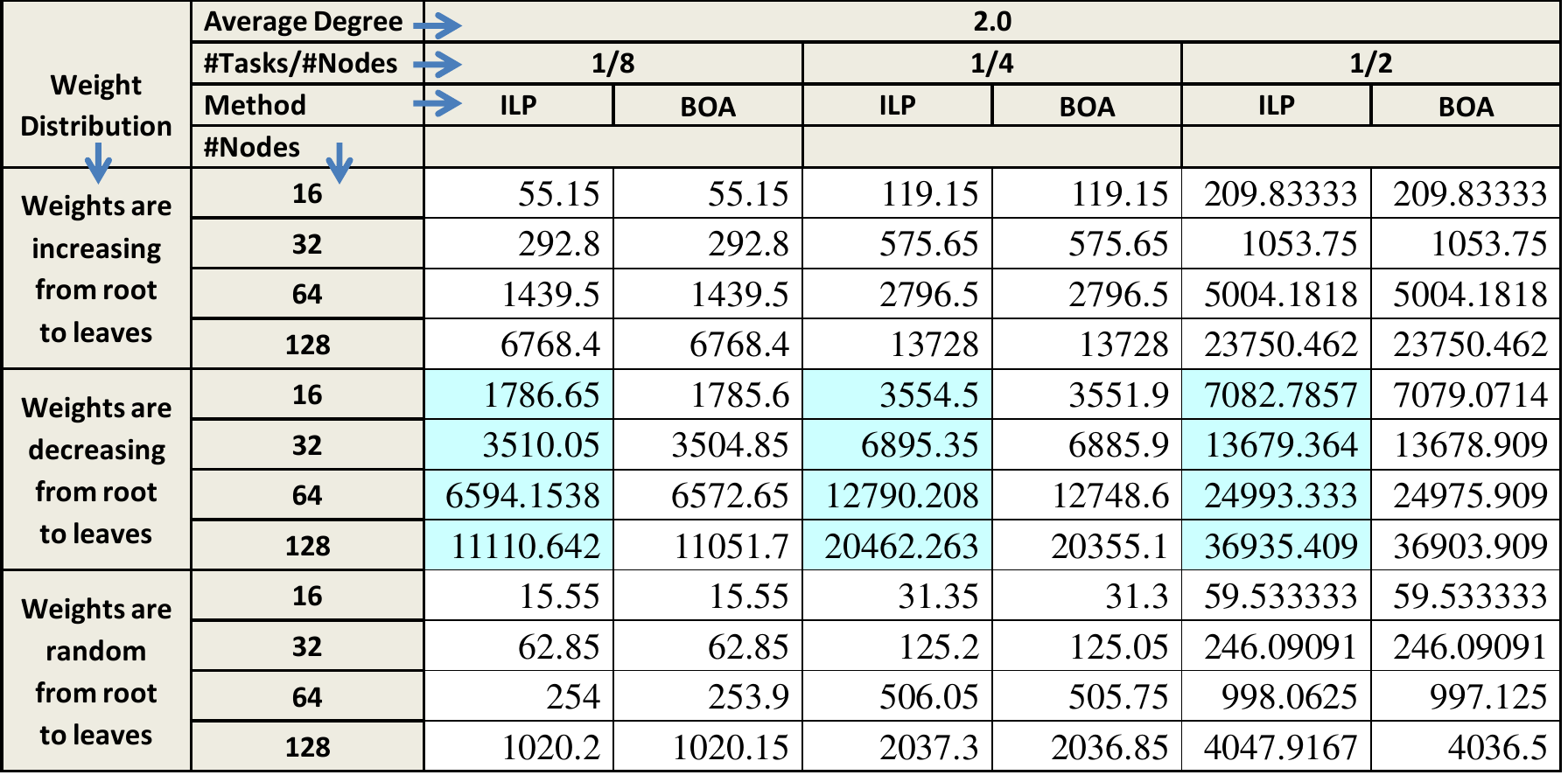}
	\caption{The solutions obtained when the average degree of a tree node
parameter is set to $2.0$. The values in the blue cells are the estimated upper bounds
obtained by the corresponding possibly fractional LP relaxation solutions.}
	\label{fig:expGoal_2.0}
\end{figure}

The tables in Figure~\ref{fig:expTime_2.0} and Figure~\ref{fig:expGoal_2.0} display
the execution times, and the solutions respectively when the \textit{Average Degree}
parameter is set to $2.0$. There are this time $4$ out of $36$ test cases where BOA
turns out to be slower than the ILP solver library, and these correspond to the cells in
Figure~\ref{fig:expTime_2.0} characterized by:
i) ${}^{\#Tasks}/{}_{\#Nodes}$ {\small = 1/8}, {\small \textit{Weight Distribution} = random},
{\small \textit{\#Nodes} = 64},
ii) ${}^{\#Tasks}/{}_{\#Nodes}$ {\small = 1/8}, {\small \textit{Weight Distribution} = random},
{\small \textit{\#Nodes} = 128},
iii) ${}^{\#Tasks}/{}_{\#Nodes}$ {\small = 1/4}, {\small \textit{Weight Distribution} = random},
{\small \textit{\#Nodes} = 64}, and
iv) ${}^{\#Tasks}/{}_{\#Nodes}$ {\small = 1/2}, {\small \textit{Weight Distribution} = increasing},
{\small \textit{\#Nodes} = 16}.
An inspection of the respective cells corresponding to these test cases in both
Figure~\ref{fig:expNoOfLPCalls} and Figure~\ref{fig:expGoal_2.0} confirms
once more that BOA finds solutions with optimal or almost optimal values in
exactly one iteration making a single LP call. As a result, the previous analysis
stating that ILP performs very fast for the instances whose LP formulations
also return integer assignments still holds.

BOA always achieves optimal or very close to optimal solutions as shown in Figure~\ref{fig:expGoal_2.0}.
For example, when ${}^{\#Tasks}/{}_{\#Nodes}=1/2$ for a $128$-node tree, and the weights
are randomly distributed among all nodes, the ILP produces the optimal goal value as $4047.9167$
and BOA heuristic generates $4036.5$. This is one of the cases with the largest difference between
the optimal solution and our heuristic solution. Even in this case, the difference between the two solutions
is much less than $1\%$. For some cases where we have used LP relaxation solutions as upper bounds
instead of ILP, the differences are slightly higher. For example, when ${}^{\#Tasks}/{}_{\#Nodes}=1/4$
for a $128$-node tree, and the weights are decreasing from the root to the leaves, the upper bound to
the optimal is $20462.263$, and BOA achieves $20355.1$. Even for this upper bound the difference is
very small. Potentially, BOA might even have the same solution as the actual optimal, or else would have
definitely achieved a closer value to the actual optimal.

\begin{figure}[htb]
	\centering
	\includegraphics[width=4.4in]{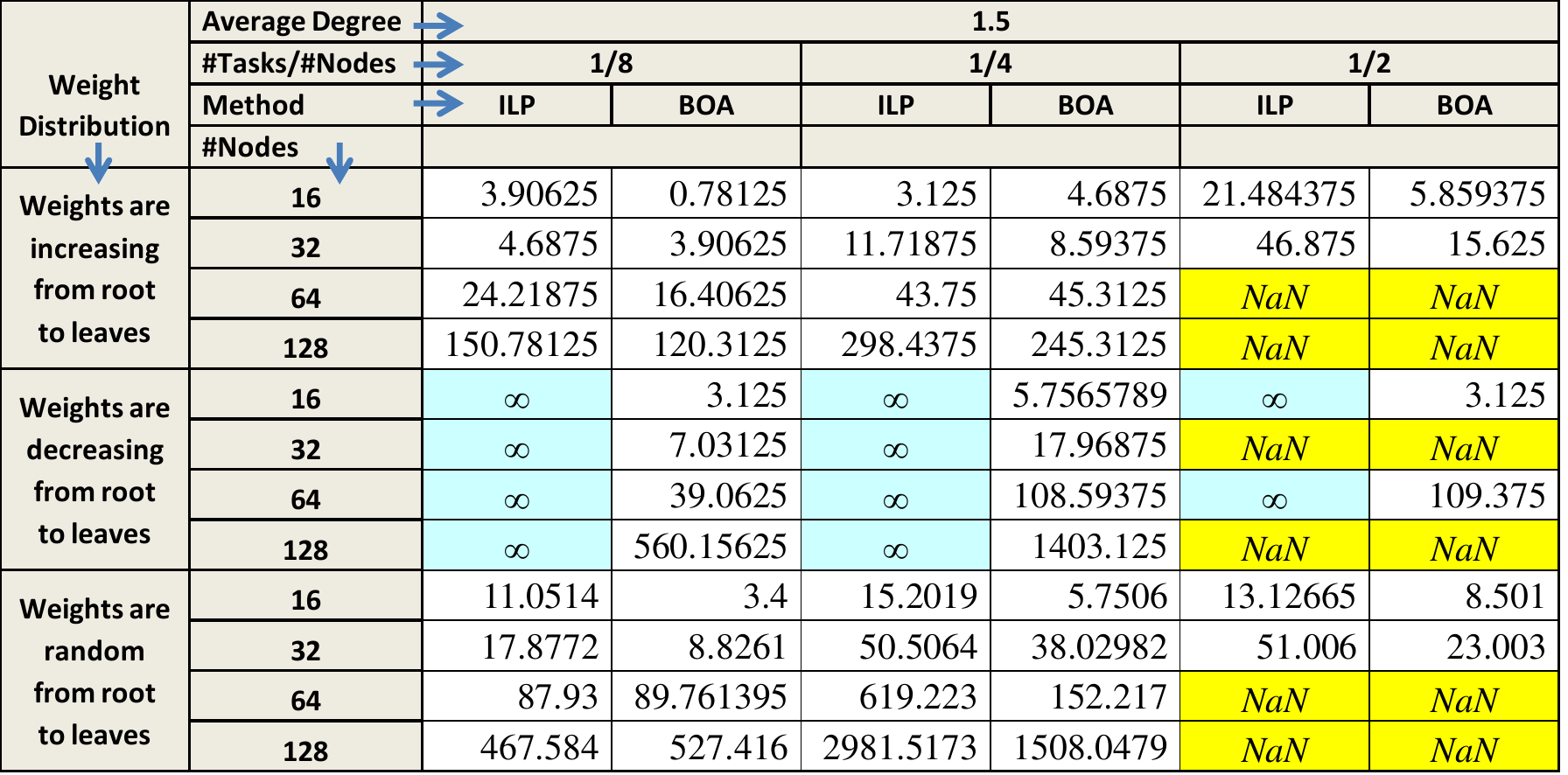}
	\caption{The execution times when the average degree of a tree node
parameter is set to $1.5$ corresponding to deep trees. While the cells in yellow
marked with the symbol \textit{NaN} represent the parameter combinations for
which there are no feasible solutions, the symbol $\infty$ in a blue cell indicates
a very large value.}
	\label{fig:expTime_1.5}
\end{figure}

Figure~\ref{fig:expTime_1.5} and Figure~\ref{fig:expGoal_1.5} display the execution
times, and the solutions respectively when the parameter representing the average
degree of a tree node is set to $1.5$ which corresponds to deep trees. In $4$ out of
the $36$ test cases presented in Figure~\ref{fig:expTime_1.5}, BOA executes longer
in figuring out a solution. The first two of these correspond to the cases where the
parameter \textit{\#Nodes} is set to either $64$ or $128$ when
${}^{\#Tasks}/{}_{\#Nodes}=1/8$ and \textit{Weight Distribution} is random.
The cells corresponding to these two test cases in Figure~\ref{fig:expNoOfLPCalls}
have both the value $1.1$. Furthermore, it is seen from the corresponding cells in
Figure~\ref{fig:expGoal_1.5} that BOA finds solutions very close to optimal.
The last two test cases correspond, however, to the combinations of parameters when
\textit{\#Nodes} is set to either $16$ or $64$ when
${}^{\#Tasks}/{}_{\#Nodes}=1/4$ and \textit{Weight Distribution} is such that
it is increasing from the root to the leaves. A quick inspection of the corresponding
cells for the last two test cases in the corresponding tables reveals that BOA found
the optimal solutions after a single LP invocation. So the prior justification is still valid.

\begin{figure}[htb]
	\centering
	\includegraphics[width=4.4in]{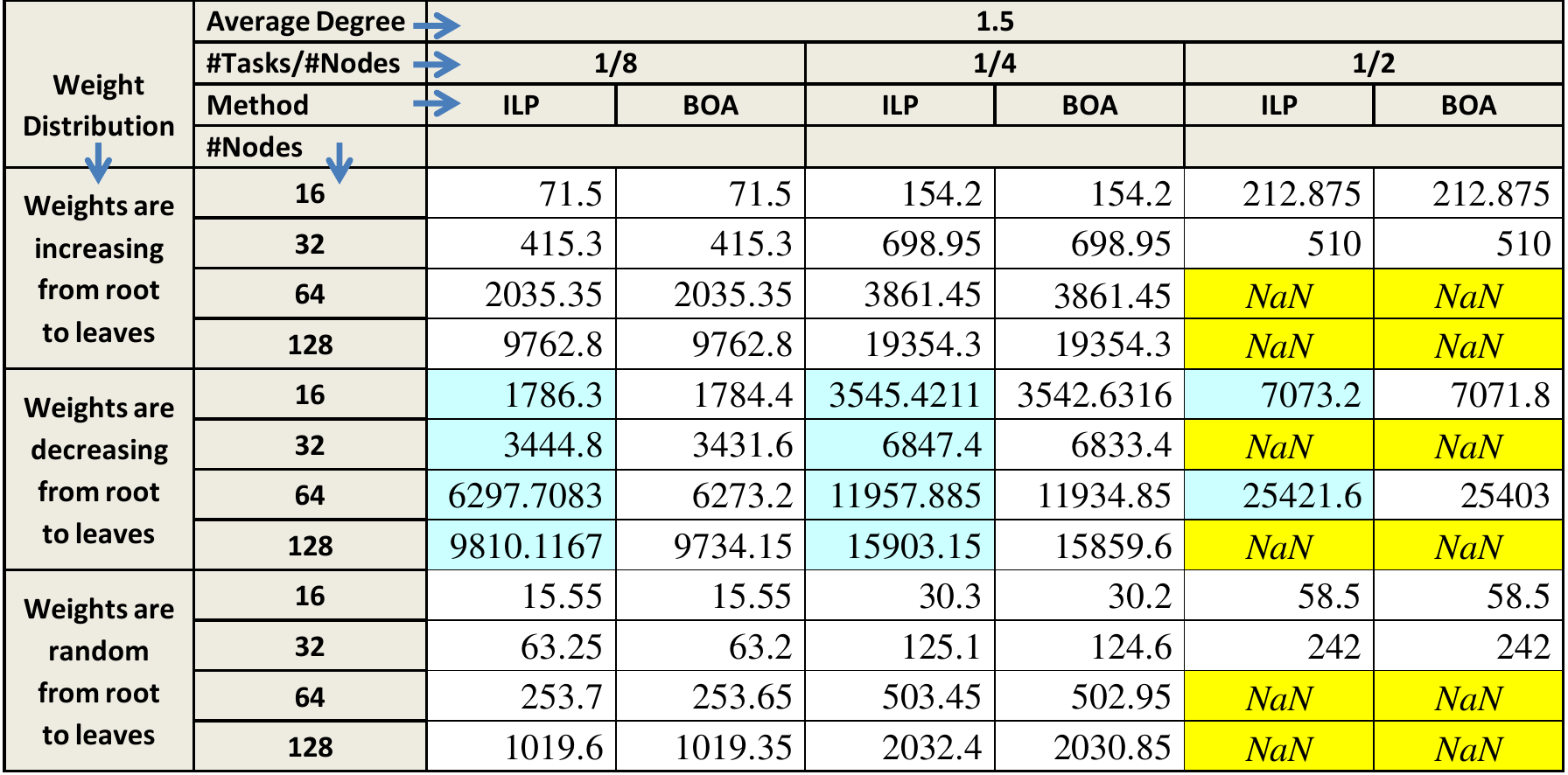}
	\caption{The solutions obtained when the average degree of a tree node
parameter is set to $1.5$. While the cells in yellow marked with the symbol \textit{NaN}
represent the parameter combinations for which there are no feasible solutions, the
values in the blue cells are the estimated upper bounds obtained by the corresponding
possibly fractional LP relaxation solutions.}
	\label{fig:expGoal_1.5}
\end{figure}

The results of the experiments show that for all cases BOA generates
goal values very close to the optimal obtained by ILP. The results are
either exactly the same, or there is a very small difference. Besides, in
the latter case, the distance to the optimal is always much less than $1\%$.
Moreover with \textit{Weight Distribution} increasing from the root to the leaves,
BOA always finds optimal solutions.

When the parameter \textit{Weight Distribution} is such that it decreasing
from the root to the leaves, it takes forever to compute the optimal by ILP
as shown by the corresponding cells marked $\infty$ throughout the tables.
Under the same setting, BOA, on the other hand, returns in polynomial time
almost optimal solutions that are within $1\%$ of even the upper bounds
obtained via the corresponding LP relaxation solution.

In only $11$ out of a total of $108$ different test cases, ILP runs faster
than ILP. All $11$ of these execution time anomalies are seen occur when
BOA discovers an almost optimal solution after at most $1$ or $1.1$ LP calls
on the average. These test cases are therefore thought to correspond most
probably to the instances that can be solved efficiently by ILP. In such a case
ILP can essentially find a solution by making only a very few LP relaxation calls
via a branch and bound algorithm. It is then easily anticipated that the additional
overhead posed by BOA leaves it behind ILP.

\section{Conclusion}

In this paper we have introduced a new version of the assignment problem,
called as $MWTM$ problem. In $MWTM$, as is the case with the standard
assignment problem, a one-to-one assignment is sought between a set of
tasks and a set of agents (nodes) to maximize the total profit (weight) value.
Moreover, there is an additional constraint in $MWTM$ preventing some
combinations of the assignments. Since agents are organized in a tree
structure representing hierarchical (agent - sub-agent) relationships, when
an agent is assigned to a task, none of its sub-agents or super-agents can
be assigned to any other task. This problem is shown to be NP-hard.
Therefore, we proposed an iterative LP-relaxation solution to it. Through
experiments we have shown that our heuristic solution is very effective,
and produces either the optimal solution, or a solution very close to the
optimal in a very reasonable time performing only a few iterations.
In most cases the solution is achieved within a single iteration.

\bibliographystyle{elsarticle-num}
\bibliography{mwtm}

\begin{thebibliography}{10}
\expandafter\ifx\csname url\endcsname\relax
  \def\url#1{\texttt{#1}}\fi
\expandafter\ifx\csname urlprefix\endcsname\relax\def\urlprefix{URL }\fi
\expandafter\ifx\csname href\endcsname\relax
  \def\href#1#2{#2} \def\path#1{#1}\fi

\bibitem{BDM09}
R.~E. Burkard, M.~Dell'Amico, S.~Martello, Assignment Problems, SIAM, 2009.

\bibitem{CKR06}
R.~Cohen, L.~Katzir, D.~Raz, An efficient approximation for the generalized
  assignment problem, Information Processing Letters 100~(4) (2006) 162--166.

\bibitem{FGMS06}
L.~Fleischer, M.~X. Goemans, V.~S. Mirrokni, M.~Sviridenko, Tight approximation
  algorithms for maximum general assignment problems, in: SODA'06, ACM, 2006,
  pp. 611--620.

\bibitem{G93}
S.~Geetha, K.~P.~K. Nair, A variation of the assignment problem, European
  Journal of Operational Research 68~(3) (1993) 422--426.

\bibitem{A95}
R.~D. Armstrong, Z.~Jin, On solving a variation of the assignment problem,
  European Journal of Operational Research 87~(1) (1995) 142--147.

\bibitem{GT10pack}
M.~Gulek, I.~H. Toroslu, A dynamic programming algorithm for tree-like weighted
  set packing problem, Information Sciences 180~(20) (2010) 3974--3979.

\bibitem{WZS13}
F.~Wang, S.~Zhou, N.~Shi, Group-to-group reviewer assignment problem, Computers
  \& Operations Research 40~(5) (2013) 1351--1362.

\bibitem{SRSP06}
T.~Shima, S.~J. Rasmussen, A.~G. Sparks, K.~M. Passino, Multiple task
  assignments for cooperating uninhabited aerial vehicles using genetic
  algorithms, Computers \& Operations Research 33~(11) (2006) 3252--3269.

\bibitem{GT10}
M.~Gulek, I.~H. Toroslu, A genetic algorithm for maximum-weighted tree matching
  problem, Applied Soft Computing 10~(4) (2010) 1127--1131.

\bibitem{J01}
K.~Jain, A factor 2 approximation algorithm for the generalized steiner network
  problem, Combinatorica 21~(1) (2001) 39--60.

\bibitem{JThesis00}
K.~Jain, Enhancing techniques in {LP} based approximation algorithms, Ph.D.
  thesis, Georgia Institute of Technology (August 2000).

\bibitem{K72}
R.~M. Karp, Reducibility among combinatorial problems, in: R.~E. Miller, J.~W.
  Thatcher (Eds.), Complexity of Computer Computations, The IBM Research
  Symposia Series, Plenum Press, New York, 1972, pp. 85--103.

\bibitem{H01}
J.~H{\aa}stad, Some optimal inapproximability results, Journal of ACM 48~(4)
  (2001) 798--859.

\bibitem{Khachiyan79}
L.~Khachiyan, A polynomial algorithm in linear programming, Soviet Math. Dokl.
  20~(1) (1979) 191--194.

\bibitem{Karmarkar84}
N.~Karmarkar, A new polynomial-time algorithm for linear programming,
  Combinatorica 4~(4) (1984) 373--395.

\end{thebibliography}

\end{document}